\documentclass{article}
\usepackage[utf8]{inputenc}
\usepackage[T1]{fontenc}
\usepackage{amssymb, amsmath, amsthm}  
\usepackage{bm}
\usepackage{mathtools}
\usepackage[final]{graphicx}
\usepackage{caption}                      
\usepackage[english]{babel}
\usepackage[a4paper, total={6.5in, 8.5in}]{geometry}
\usepackage{enumerate}
\usepackage{float}
\usepackage{verbatim}
\usepackage{tikz}
\usetikzlibrary{fit, matrix, calc, positioning, decorations.shapes, decorations.pathreplacing}
\usepackage{makeidx}
\usepackage{algpseudocode}
\usepackage{algorithm}
\usepackage{xcolor}
\usepackage{nicematrix}
\usepackage{hyperref}  
\usepackage[style=numeric, sorting=nyt]{biblatex}
\DeclareFieldFormat[article]{title}{#1}      
\DeclareFieldFormat[book]{title}{#1}         
\DeclareFieldFormat[inbook]{title}{#1}       
\DeclareFieldFormat{booktitle}{#1}           

\renewbibmacro{in:}{}                        
\DeclareFieldFormat{pages}{#1}               
\usepackage{xurl}  
\DeclareFieldFormat{doi}{\href{https://doi.org/#1}{\nolinkurl{#1}}}
\DeclareFieldFormat{url}{\href{#1}{#1}}
\DeclareNameAlias{author}{family-given}

\newtheorem{theorem}{Theorem}[section]

\newtheorem{corollary}[theorem]{Corollary}
\newtheorem{proposition}[theorem]{Proposition}
\newtheorem{lemma}[theorem]{Lemma}

\DeclareMathOperator*{\argmin}{arg\,min}
\newcommand{\bbp}{\mathbb{P}}
\newcommand{\bE}{\mathbb{E}}
\newcommand{\bbn}{\mathbb{N}}

\newcommand{\mT}{\mathcal{T}}
\newcommand{\mX}{\mathcal{X}}
\newcommand{\mF}{\mathcal{F}}
\newcommand{\Fmat}{\mathbb{F}\text{-matrix}}
\newcommand{\Fmats}{\mathbb{F}\text{-matrices}}

\newcommand{\var}{\mathbb{V}\text{ar}}
\newcommand{\cov}{\mathbb{C}\text{ov}}
\newcommand{\transpose}{\mathsf{T}}
\newcommand{\bx}{\mathbf{x}}
\providecommand{\keywords}[1]
{
  \small	
  \textbf{Keywords:} #1
}
\providecommand{\mXC}[1]
{
  \small	
  \textbf{Mathematics Subject Classification:} #1
}
\title{Markov embedding of ranked unlabelled evolutionary trees and its applications}
\author{Lasse Thorup Fallesen$^{1}$, Simon Pauli$^{2}$, Elisabeth Sommer James$^{1}$, Lars Nørvang Andersen$^{1}$,\\ Asger Hobolth$^{1}$ \\
        \small $^{1}$Department of Mathematics, Aarhus University, Ny Munkegade 118, DK-8000 Aarhus C, Denmark \\
        \small $^{2}$Institute of Applied Statistics, Johannes Kepler University, Linz, Austria \\
}
\date{}
\begin{document}
\maketitle
\begin{abstract}\noindent
Rooted bifurcating trees are mathematical objects used to model evolutionary relationships and arise naturally in both coalescent theory and phylogenetics. Recent numerical representations of tree topologies, known as $\Fmats$, allow for summarizing a sample of trees via Frechét means and provide new measures of tree balance. However, the number of ranked unlabelled trees grows super-exponentially with the number of leaves. This makes computation intensive and current methods rely on mixed integer programming and simulation-based methods. Moreover, $\Fmats$ are difficult to interpret, and their distribution is only described in terms of first- and second-order moments under neutral branching. In this paper, we introduce a Markov chain embedding of ranked and unlabelled trees that drastically decreases the size of the state space. Leveraging this embedding, we develop an algorithm that efficiently computes all Fréchet means and use discrete phase-type theory to obtain the joint distribution of tree balance indices. We also use discrete phase-type theory to generalize previous results regarding moments of $\Fmats$ to arbitrary order for any time homogeneous and bifurcating coalescent model. Using this framework, we construct three tests for neutrality and demonstrate their improved power compared to previous methods on simulated data. 
\end{abstract}

\noindent\keywords{Binary Tree  $\cdot$ Unlabelled tree $\cdot$ Markov chain $\cdot$ Phase-type distributions $\cdot$ Tree balance}\medbreak

\noindent\mXC{Primary: 60J10 $\cdot$ Secondary: 60J90, 92D15, 60J22}
\section{Introduction}
Rooted and bifurcating trees provide a natural representation of evolutionary and ancestral relationships between species in phylogenetics and individuals within the same species in population genetics. They encode both forward-in-time branching processes and the order in which $n$ lineages coalesce backward-in-time. The branch lengths and topology of a tree offer insights into past evolutionary dynamics and genetic variation \cite{ColjinPlazzotta, Kayondo}. While branch lengths describe the coalescent times, the topology of a tree captures the relative ordering of events as well as tree balance.

Since the introduction of coalescent theory in 1982 \cite{Coalescent}, trees with associated coalescent times have been studied extensively. However, the topology of the ranked and unlabelled tree alone can influence genetic diversity and carries information about evolutionary patterns and tree balance that are independent of labeling and branch lengths. Such structural features can be used to test hypotheses about deviance from neutrality \cite{SamPal, DrummondSuchard, ColjinPlazzotta}, motivating the study of these topologies. 

Ranked and unlabelled trees arise naturally in the study of ancestral relationships among exchangeable objects, where the identity or \textit{label} of each lineage or taxon is irrelevant, but the order of coalescent or branching events is informative. For example, they can be used to characterize heterogeneity across a sample of evolutionary histories of cells in cancer tumours \cite{SamPal}. We may want to characterize the branching pattern for a specific type of cancer or compare evolutionary patterns across different cancer types. 

Recently, a numeric representation called $\Fmats$ \cite{SamPal, FmatPaper} has been proposed to encode ranked unlabelled trees. This representation provides a convenient framework for comparing tree topologies and summarizing a sample of ranked and unlabelled trees. However, the number of $\Fmats$ remains large, growing superexponentially with the number of taxa. This makes exact calculations computationally expensive.

In this paper, we address this complexity by introducing a Markov embedding of ranked and unlabelled trees. In Section $2$, we construct the embedding using the $\Fmat$ representation and demonstrate how these arise naturally in a coalescent framework when branch lengths are ignored. The formulation yields a state space that grows according to $((1+\sqrt{5})/2)^n$, which is many orders of magnitude smaller than the full space of ranked unlabelled trees. This provides a tractable probabilistic framework for working with ranked tree shapes, which we capitalize on in the remainder of the paper.

In Section $3$, we exploit the reduced state space along with the Markov property to develop a new algorithm to efficiently summarize a sample of tree shapes or a distribution on the space thereof. In Section $4$, we introduce new balance indices tailored to ranked tree shapes based on the $\Fmat$ representation. In Section $5$, we show how the embedding also enables the use of discrete phase-type theory to derive distributional properties of functionals of the individual states. This yields matrix-analytic expressions for all moments of $\Fmat$ entries, and also generalizes existing results \cite{SamPal} beyond the Kingman model. In Section $6$, we exploit the feed-forward structure of the chain to develop methods that avoid expensive matrix operations. In Section $7$, we use the phase-type framework to develop three novel and powerful tests of neutrality. Lastly, in Section $8$, we discuss the relation of the embedding to the block-counting process and extensions of the framework.

We apply the methods to the Kingman coalescent and to trees sampled from the Blum-Fran\c{c}ois $\beta$-splitting model \cite{Blum-Francois}. Compared with existing results \cite{SamPal}, we observe substantial improvements in both computational efficiency and statistical power. An implementation of the methods developed here is done in \texttt{R} \cite{R} and is available at \href{https://github.com/LasseThorup1/RankedCoalescent}{https://github.com/LasseThorup1/RankedCoalescent}. All figures in the paper have been made using data available in the repository.\vspace{-10pt}

\section{Markov embedding of ranked unlabelled trees}
The ranked unlabelled trees analyzed in this paper are rooted binary trees with ordered internal nodes. The branching events are ranked, starting from $2$, and describes the total number of branches immediately after each event. Figure \ref{fig:Figure-Fmat} gives an overview of our Markov embedding and interpretation of the $\Fmat$ notation. A representation of all $5$ distinct ranked tree topologies with $n = 5$ leaves is provided in Figure \ref{fig:Figure-Fmat}A. While this is a forward-in-time perspective, the same structures can be viewed backward-in-time as a coalescent process. In regard to the latter, the rankings decrease starting from $n$ and denote the number of lineages immediately prior to each coalescent event. The process is unlabelled in the sense that lineages (backward-in-time) or leaves (forward-in-time) are unlabelled. We assume that the trees are isochronous, implying that no lineage goes extinct, or equivalently, that all leaves are sampled simultaneously. Branch lengths are not considered here as our focus lies solely on the ranked topology of the process.

Let $\mT_n$ denote the space of ranked unlabelled trees with $n$ leaves. Each tree $T \in \mT_n$ can be uniquely represented by a lower-triangular, integer-valued  $\Fmat$ of dimension $n-1$ \cite{FmatPaper}. The diagonal elements $F_{ii} = i + 1$ represent the number of branches immediately following the $i$th branching event, and the off-diagonal elements $F_{ij}, 1 \leq j < i \leq n-1$ represent the number of extant branches immediately following the $j$th branching event that do not bifurcate prior to the $(i+1)$th branching event. Figure \ref{fig:Figure-Fmat}B shows the $\Fmats$ corresponding to the trees in Figure \ref{fig:Figure-Fmat}A.

In \cite{SamPal} the space of these $\Fmats$ is introduced via a set of linear constraints. In contrast, the Markov embedding illustrated in Figure \ref{fig:Figure-Fmat}C bypasses these constraints by representing trees as sample paths of a Markov chain, and in Figure \ref{fig:Figure-Fmat}D we demonstrate how the $\Fmats$ naturally arise as a sequence of coalescent events corresponding to the transitions in the chain. We now formally introduce the embedding.\vspace{-5pt}

\subsection{The Markov embedding of the ranked coalescent}
As introduced above, $\Fmats$ encode the forward-in-time branching process. However, a bifurcating tree can also be viewed as the outcome of a coalescent process, backward-in-time, where we consider each branch as a lineage. Here we adopt the latter viewpoint and describe a Markov chain whose sample paths are in one-to-one correspondence with the space of $\Fmats$, and we refer to this Markov chain as the \textit{ranked coalescent}.

Let $\{X_t\}_{t=0,\cdots,n-1}$ denote the ranked coalescent. For $t = 0,\cdots, n-2$, $X_t$ corresponds to the $(n-1-t)$th column in an appropriate $\Fmat$, while $X_{n-1} = \text{MRCA}$ represents the absorbing state. By construction, we will always have $X_0 = (0,\cdots,0,n)^\transpose, X_1 = (0,\cdots,n-1,n-2)^\transpose$. We denote the state space of the ranked coalescent by $\mX_n$, and each non-absorbing state in $\mX_n$ is denoted by $\bx = (x_1,\cdots,x_{n-1})^\transpose$.

Below, we illustrate by example how the columns of the $\Fmats$ (states in the Markov chain) can be interpreted as partial trees and how the chain transitions from one column to a column to its immediate left by coalescing two lineages present in the current state. We do so by considering the $11$-leaved tree along with its corresponding $\Fmat$ shown in Figure \ref{fig:twoTrees}.
\begin{figure} [!bt]
    \centering
    \scalebox{0.75}{
    \begin{tikzpicture}[node distance=4cm]
        \node (Tree1) {\begin{tikzpicture} 
        \draw (0, 0.5) -- (0, 0);
        \draw (-1, 0) -- (1, 0);
        \draw (-1, 0) -- (-1, -0.5);
        \draw (1, 0) -- (1, -2);
        \draw (-1.5, -0.5) -- (-0.5, -0.5);
        \draw (-1.5, -0.5) -- (-1.5, -1);
        \draw (-2, -1) -- (-1, -1);
        \draw (-2, -1) -- (-2, -1.5);
        \draw (-1, -1) -- (-1, -2);
        \draw (-2.5, -1.5) -- (-1.5, -1.5);
        \draw (-2.5, -1.5) -- (-2.5, -2);
        \draw (-1.5, -1.5) -- (-1.5, -2);
        \draw (-0.5, -0.5) -- (-0.5, -2);
        \node[text width = 1cm] at (0.2,0.2) {$2$};
        \node[text width = 1cm] at (-0.8,-0.3) {$3$};
        \node[text width = 1cm] at (-1.3,-0.8) {$4$};
        \node[text width = 1cm] at (-1.8,-1.3) {$5$};
        \node[text width = 1cm] at (-0.2,0.75) {\textbf{Absorbing}};
    \end{tikzpicture}} ;
    \node (Tree2) [right of=Tree1] {\begin{tikzpicture}
        \draw [red](0, 0.5) -- (0, 0);
        \draw [red](-1, 0) -- (1, 0);
        \draw [red](-1, 0) -- (-1, -0.5);
        \draw [red](1, 0) -- (1, -2);
        \draw [red](-1.5, -0.5) -- (0, -0.5);
        \draw [red](-1.5, -0.5) -- (-1.5, -1);
        \draw [red](-2, -1) -- (-1, -1);
        \draw [red](-2, -1) -- (-2, -2);
        \draw [red](-1, -1) -- (-1, -2);
        \draw [red](0, -0.5) -- (0, -1.5);
        \draw [red](-0.5, -1.5) -- (0.5, -1.5);
        \draw [red](-0.5, -1.5) -- (-0.5, -2);
        \draw [red](0.5, -1.5) -- (0.5, -2);
        \node[text width = 1cm, red] at (0.2,0.2) {$2$};
        \node[text width = 1cm, red] at (-0.8,-0.3) {$3$};
        \node[text width = 1cm, red] at (-1.3,-0.8) {$4$};
        \node[text width = 1cm, red] at (0.2,-1.3) {$5$};
        \node[text width = 1cm, red] at (-0.2,0.75) {\textbf{Absorbing}};
    \end{tikzpicture}} ;
    \node (Tree3) [right of=Tree2] {\begin{tikzpicture} 
        \draw [blue] (0, 0.5) -- (0, 0);
        \draw [blue](-1, 0) -- (1, 0);
        \draw [blue](-1, 0) -- (-1, -0.5);
        \draw [blue](1, 0) -- (1, -1.5);
        \draw [blue](-1.5, -0.5) -- (-0.5, -0.5);
        \draw [blue](-1.5, -0.5) -- (-1.5, -1);
        \draw [blue](-0.5, -0.5) -- (-0.5, -2);
        \draw [blue](-2, -1) -- (-1, -1);
        \draw [blue](-2, -1) -- (-2, -2);
        \draw [blue](-1, -1) -- (-1, -2);
        \draw [blue](0.5, -1.5) -- (1.5, -1.5);
        \draw [blue](0.5, -1.5) -- (0.5, -2);
        \draw [blue](1.5, -1.5) -- (1.5,-2);
        \node[text width = 1cm, blue] at (0.2,0.2) {$2$};
        \node[text width = 1cm, blue] at (-0.8,-0.3) {$3$};
        \node[text width = 1cm, blue] at (-1.3,-0.8) {$4$};
        \node[text width = 1cm, blue] at (1.2,-1.3) {$5$};
        \node[text width = 1cm, blue] at (-0.2,0.75) {\textbf{Absorbing}};
    \end{tikzpicture}};

    \node (Tree4) [right of=Tree3] {\begin{tikzpicture}
        \draw [orange](0, 0.5) -- (0, 0);
        \draw [orange](-1, 0) -- (1, 0);
        \draw [orange](-1, 0) -- (-1, -0.5);
        \draw [orange](1, 0) -- (1, -1);
        \draw [orange](-1.5, -0.5) -- (-0.5, -0.5);
        \draw [orange](-1.5, -0.5) -- (-1.5, -2);
        \draw [orange](-0.5, -0.5) -- (-0.5, -2);
        \draw [orange](0.5, -1) -- (1.5, -1);
        \draw [orange](0.5, -1) -- (0.5, -1.5);
        \draw [orange](1.5, -1) -- (1.5, -2);
        \draw [orange](0, -1.5) -- (1, -1.5);
        \draw [orange](0, -1.5) -- (0, -2);
        \draw [orange](1, -1.5) -- (1, -2);
        \node[text width = 1cm, orange] at (0.2,0.2) {$2$};
        \node[text width = 1cm, orange] at (-0.8,-0.3) {$3$};
        \node[text width = 1cm, orange] at (1.2,-0.8) {$4$};
        \node[text width = 1cm, orange] at (0.7,-1.3) {$5$};
        \node[text width = 1cm, orange] at (-0.2,0.75) {\textbf{Absorbing}};
    \end{tikzpicture}};

    \node (Tree5) [right of=Tree4] {\begin{tikzpicture} 
        \draw [purple](-1, 0) -- (1, 0);
        \draw [purple](0, 0.5) -- (0, 0);
        \draw [purple](-1, 0) -- (-1, -0.5);
        \draw [purple](1, 0) -- (1, -1);
        \draw [purple](-1.5, -0.5) -- (-0.5, -0.5);
        \draw [purple](-1.5, -0.5) -- (-1.5, -1.5);
        \draw [purple](-0.5, -0.5) -- (-0.5, -2);
        \draw [purple](0.5, -1) -- (1.5, -1);
        \draw [purple](0.5, -1) -- (0.5, -2);
        \draw [purple](1.5, -1) -- (1.5, -2);
        \draw [purple](-2, -1.5) -- (-1, -1.5);
        \draw [purple](-2, -1.5) -- (-2, -2);
        \draw [purple](-1, -1.5) -- (-1, -2); 
        \node[text width = 1cm, purple] at (0.2,0.2) {$2$};
        \node[text width = 1cm, purple] at (-0.8,-0.3) {$3$};
        \node[text width = 1cm, purple] at (1.2,-0.8) {$4$};
        \node[text width = 1cm, purple] at (-1.3,-1.3) {$5$};
        \node[text width = 1cm, purple] at (-0.2,0.75) {\textbf{Absorbing}};
    \end{tikzpicture}};

    \node (F1) [below=1.2cm of Tree1] {$F_1 = \begin{bmatrix}
        2&0&0&0\\
        1&3&0&0\\
        1&2&4&0\\
        1&2&3&5
    \end{bmatrix}$};

    \node (F2) [right of=F1, red] {$F_2 = \begin{bmatrix}
        2&0&0&0\\
        1&3&0&0\\
        1&2&4&0\\
        1&1&3&5
    \end{bmatrix}$};

    \node (F3) [right of=F2, blue] {$F_3 = \begin{bmatrix}
        2&0&0&0\\
        1&3&0&0\\
        1&2&4&0\\
        0&1&3&5
    \end{bmatrix}$};

    \node (F4) [right of=F3, orange] {$F_4 = \begin{bmatrix}
        2&0&0&0\\
        1&3&0&0\\
        0&2&4&0\\
        0&2&3&5
    \end{bmatrix}$};

    \node (F5) [right of=F4, purple] {$F_5 = \begin{bmatrix}
        2&0&0&0\\
        1&3&0&0\\
        0&2&4&0\\
        0&1&3&5
    \end{bmatrix}$};

    \node (Embedding) at ([shift={(3cm,-6cm)}] F1) {\scalebox{0.9}{\begin{tikzpicture}[->]
    \node [label = above:state 1] at(-2,0) (state 1){$\begin{pmatrix}
        0\\0\\0\\5
    \end{pmatrix}$};
    
    \node [label = above:state 2] at(0,0) (state 2){$\begin{pmatrix}
        0\\0\\4\\3
    \end{pmatrix}$};
    
    \node [label = below:state 4] at(2, -1.5) (state 4){$\begin{pmatrix}
        0\\3\\2\\1
    \end{pmatrix}$};
    
    \node [label = above:state 3] at(2, 1.5) (state 3){$\begin{pmatrix}
        0\\3\\2\\2
    \end{pmatrix}$};
    
    \node [label = right:state 7] at(4, -3) (state 7){$\begin{pmatrix}
        2\\1\\0\\0
    \end{pmatrix}$};
    
    \node [label = right:state 6] at(4, 0) (state 6){$\begin{pmatrix}
        2\\1\\1\\0
    \end{pmatrix}$};
    
    \node [label = right:state 5] at(4, 3)(state 5){$\begin{pmatrix}
        2\\1\\1\\1
    \end{pmatrix}$};

    \node at (7.5, 0)(absorbing) {MRCA};

    \node at (4, 4.5){\textbf{column 1}};

    \node at (2, 3){\textbf{column 2}};

    \node at (0, 1.5){\textbf{column 3}};

    \node at (-2, 1.5){\textbf{column 4}};
    \path (state 1) edge (state 2);
    \path (state 2) edge (state 3);
    \path (state 2) edge (state 4);
    \path [line width =2.5pt](state 3) edge (state 5);
    \path [line width =2.5pt,orange] (state 3) edge (state 7);
    \path [line width =2.5pt,red] (state 4) edge (state 5);
    \path [line width =2.5pt,blue] (state 4) edge (state 6);
    \path [line width =2.5pt, purple] (state 4) edge (state 7);
    \path (state 5) edge (absorbing);
    \path (5.8, 0) edge (6.8,0);
    \path (state 7) edge (absorbing);
    \coordinate (dm1) at (4,4);
    \coordinate (dm2) at (4,-4);
    \node[rectangle,draw,minimum width=1cm, color = red] [fit = (dm1) (dm2)] {};
    
    \coordinate (dm3) at (2,2.4);
    \coordinate (dm4) at (2,-2.4);
    \node[rectangle,draw,minimum width=1cm, color = red] [fit = (dm3) (dm4)] {};
    
    \coordinate (dm5) at (0,0.9);
    \coordinate (dm6) at (0,-0.9);
    \node[rectangle,draw,minimum width=1cm, color = red] [fit = (dm5) (dm6)] {};
    
    \coordinate (dm7) at (-2,0.9);
    \coordinate (dm8) at (-2,-0.9);
    \node[rectangle,draw,minimum width=1cm, color = red] [fit = (dm7) (dm8)] {};
    \end{tikzpicture}}};

    \node (StateSpaceComp) at ([shift={(10cm,0cm)}]Embedding) {\scalebox{1}{\begin{tikzpicture}
        \node (F1) {$F_3 = \begin{bmatrix}
            2&0&0&0\\
            1&3&0&0\\
            1&2&4&0\\
            0&1&3&5
        \end{bmatrix}$};

        \node[anchor=west] (X0) [below=0cm of F1] {\scalebox{0.7}{$X_0 = \begin{pmatrix}0\\0\\0\\5\end{pmatrix} = \begin{pmatrix}
            0\\0\\0\\1
        \end{pmatrix}+\begin{pmatrix}
            0\\0\\0\\1
        \end{pmatrix}+\begin{pmatrix}
            0\\0\\0\\1
        \end{pmatrix}+\begin{pmatrix}
            0\\0\\0\\1
        \end{pmatrix}+\begin{pmatrix}
            0\\0\\0\\1
        \end{pmatrix}$}}; 

        \node[anchor=west] (X1) [below=0cm of X0] {\scalebox{0.7}{$X_1 = \begin{pmatrix}0\\0\\4\\3\end{pmatrix} = \begin{pmatrix}
            0\\0\\1\\1
        \end{pmatrix}+\begin{pmatrix}
            0\\0\\1\\1
        \end{pmatrix}+\begin{pmatrix}
            0\\0\\1\\1
        \end{pmatrix}+\begin{pmatrix}
            0\\0\\0\\1
        \end{pmatrix}$}};

        \node[anchor=west] (X2) [below=0cm of X1] {\scalebox{0.7}{$X_2 = \begin{pmatrix}0\\3\\2\\1\end{pmatrix} = \begin{pmatrix}
            0\\1\\0\\0
        \end{pmatrix}+\begin{pmatrix}
            0\\1\\1\\1
        \end{pmatrix}+\begin{pmatrix}
            0\\1\\1\\0
        \end{pmatrix}$}};

        \node[anchor=west] (X3) [below=0cm of X2] {\scalebox{0.7}{$X_3 = \begin{pmatrix}2\\1\\1\\0\end{pmatrix} = \begin{pmatrix}
            1\\0\\0\\0
        \end{pmatrix}+\begin{pmatrix}
            1\\1\\1\\0
        \end{pmatrix}$}};

        \node (Tree0) [right= 1cm of X0] {\scalebox{0.45}{\begin{tikzpicture}
            \draw (0, 0.5) -- (0,0);
            \draw (-1, 0) -- (1,0);
            \draw (-1, 0) -- (-1, -0.5);
            \draw (-1.5, -0.5) -- (-0.5, -0.5);
            \draw (-1.5, -0.5) -- (-1.5, -1);
            \draw (-2, -1) -- (-1, -1);
            \draw (-2, -1) -- (-2, -2);
            \draw (-1, -1) -- (-1, -2);
            \draw (-0.5, -0.5) -- (-0.5, -2);
            \draw (1, 0) -- (1, -1.5);
            \draw (0.5, -1.5) -- (1.5,-1.5);
            \draw [red,line width = 1mm](0.5, -1.5) -- (0.5, -2);
            \draw [red,line width = 1mm](1.5, -1.5) -- (1.5, -2);
            \draw [red,line width = 1mm](-2,-1.5) -- (-2,-2);
            \draw [red,line width = 1mm](-1,-1.5) -- (-1,-2);
            \draw [red,line width = 1mm](-0.5,-1.5) -- (-0.5,-2);
            \node[text width = 1cm] at (-0.4,0.2) {$2$};
            \node[text width = 1cm] at (-1.4,-0.3) {$3$};
            \node[text width = 1cm] at (-1.9,-0.8) {$4$};
            \node[text width = 1cm] at (0.6,-1.3) {$5$};
        \end{tikzpicture}}};

        \node (Tree1) [below= 0cm of Tree0] {\scalebox{0.45}{\begin{tikzpicture}
            \draw (0, 0.5) -- (0,0);
            \draw (-1, 0) -- (1,0);
            \draw (-1, 0) -- (-1, -0.5);
            \draw (-1.5, -0.5) -- (-0.5, -0.5);
            \draw (-1.5, -0.5) -- (-1.5, -1);
            \draw (-2, -1) -- (-1, -1);
            \draw (-2, -1) -- (-2, -2);
            \draw (-1, -1) -- (-1, -2);
            \draw (-0.5, -0.5) -- (-0.5, -2);
            \draw (1, 0) -- (1, -1.5);
            \draw (0.5, -1.5) -- (1.5,-1.5);
            \draw (0.5, -1.5) -- (0.5, -2);
            \draw (1.5, -1.5) -- (1.5, -2);
            \draw[red,line width = 1mm] (-2, -1) -- (-2,-2);
            \draw[red,line width = 1mm] (-1, -1) -- (-1,-2);
            \draw[red,line width = 1mm] (-0.5,-1) -- (-0.5,-2);
            \draw[red,line width = 1mm] (1,-1) -- (1,-1.5);
            \node[text width = 1cm] at (0.2,0.43) {$2$};
            \node[text width = 1cm] at (-0.8,-0.06) {$3$};
            \node[text width = 1cm] at (-1.3,-0.56) {$4$};
            \node[text width = 1cm] at (1.2,-1.08) {$5$};
        \end{tikzpicture}}};

        \node (Tree2) [below= 0cm of Tree1] {\scalebox{0.45}{\begin{tikzpicture}
            \draw (0, 0.5) -- (0,0);
            \draw (-1, 0) -- (1,0);
            \draw (-1, 0) -- (-1, -0.5);
            \draw (-1.5, -0.5) -- (-0.5, -0.5);
            \draw (-1.5, -0.5) -- (-1.5, -1);
            \draw (-2, -1) -- (-1, -1);
            \draw (-2, -1) -- (-2, -2);
            \draw (-1, -1) -- (-1, -2);
            \draw (-0.5, -0.5) -- (-0.5, -2);
            \draw (1, 0) -- (1, -1.5);
            \draw (0.5, -1.5) -- (1.5,-1.5);
            \draw (0.5, -1.5) -- (0.5, -2);
            \draw (1.5, -1.5) -- (1.5, -2);
            \draw[red,line width = 1mm] (-1.5,-0.5) -- (-1.5,-1);
            \draw[red,line width = 1mm] (-0.5,-0.5) -- (-0.5,-2);
            \draw[red,line width = 1mm] (1,-0.5) -- (1,-1.5);
            \node[text width = 1cm] at (0.2,0.43) {$2$};
            \node[text width = 1cm] at (-0.8,-0.06) {$3$};
            \node[text width = 1cm] at (-1.3,-0.56) {$4$};
            \node[text width = 1cm] at (1.2,-1.08) {$5$};
        \end{tikzpicture}}};

        \node (Tree3) [below= 0cm of Tree2] {\scalebox{0.45}{\begin{tikzpicture}
            \draw (0, 0.5) -- (0,0);
            \draw (-1, 0) -- (1,0);
            \draw (-1, 0) -- (-1, -0.5);
            \draw (-1.5, -0.5) -- (-0.5, -0.5);
            \draw (-1.5, -0.5) -- (-1.5, -1);
            \draw (-2, -1) -- (-1, -1);
            \draw (-2, -1) -- (-2, -2);
            \draw (-1, -1) -- (-1, -2);
            \draw (-0.5, -0.5) -- (-0.5, -2);
            \draw (1, 0) -- (1, -1.5);
            \draw (0.5, -1.5) -- (1.5,-1.5);
            \draw (0.5, -1.5) -- (0.5, -2);
            \draw (1.5, -1.5) -- (1.5, -2);
            \draw[red,line width = 1mm] (-1, 0) -- (-1, -0.5);
            \draw[red,line width = 1mm] (1, 0) -- (1, -1.5);
            \node[text width = 1cm] at (0.2,0.43) {$2$};
            \node[text width = 1cm] at (-0.8,-0.06) {$3$};
            \node[text width = 1cm] at (-1.3,-0.56) {$4$};
            \node[text width = 1cm] at (1.2,-1.08) {$5$};
        \end{tikzpicture}}};

        \node (WholeTree) [right=0cm of F1] {\scalebox{0.6}{\begin{tikzpicture}
            \draw (0, 0.5) -- (0,0);
            \draw (-1, 0) -- (1,0);
            \draw (-1, 0) -- (-1, -0.5);
            \draw (-1.5, -0.5) -- (-0.5, -0.5);
            \draw (-1.5, -0.5) -- (-1.5, -1);
            \draw (-2, -1) -- (-1, -1);
            \draw (-2, -1) -- (-2, -2);
            \draw (-1, -1) -- (-1, -2);
            \draw (-0.5, -0.5) -- (-0.5, -2);
            \draw (1, 0) -- (1, -1.5);
            \draw (0.5, -1.5) -- (1.5,-1.5);
            \draw (0.5, -1.5) -- (0.5, -2);
            \draw (1.5, -1.5) -- (1.5, -2);
            \node[text width = 1cm] at (-0.4,0.2) {$2$};
            \node[text width = 1cm] at (-1.4,-0.3) {$3$};
            \node[text width = 1cm] at (-1.9,-0.8) {$4$};
            \node[text width = 1cm] at (0.6,-1.3) {$5$};
        \end{tikzpicture}}};
        
    \end{tikzpicture}}}; 
    
    \node (HeaderNode) at ([shift={(-7cm,5cm)}] Tree1) {};

    \node[anchor=west] at ([shift={(0cm,0cm)}] HeaderNode)  {\fontsize{18}{22}\selectfont \textbf{A:}  All ranked trees $T \in \mT_5$ with $n = 5$ leaves};

    \node[anchor=west] at ([shift={(0cm,-4.225cm)}] HeaderNode) {\fontsize{18}{22}\selectfont \textbf{B:} $\Fmat$ encoding of trees};

    \node[anchor=west] at ([shift={(0cm,-7.3cm)}] HeaderNode) {\fontsize{18}{22}\selectfont \textbf{C:} Markov embedding with n = 5};

    \node[anchor=west] at ([shift={(10cm, -7.3cm)}] HeaderNode) {\fontsize{18}{22}\selectfont \textbf{D:} Coalescent interpretation};
    \end{tikzpicture}}
    \caption{{\bf A:} All ranked tree shapes with $n = 5$ leaves. {\bf B:} The corresponding $\Fmat$ representations. {\bf C:} The transition diagram of the Markov embedding, where arrows from column 2 to column 1 are colored according to the corresponding tree in {\bf A}. {\bf D:} Interpretation of $F_3$ and its corresponding tree in the embedding.}
    \label{fig:Figure-Fmat}
    \vspace{-10pt}
\end{figure}
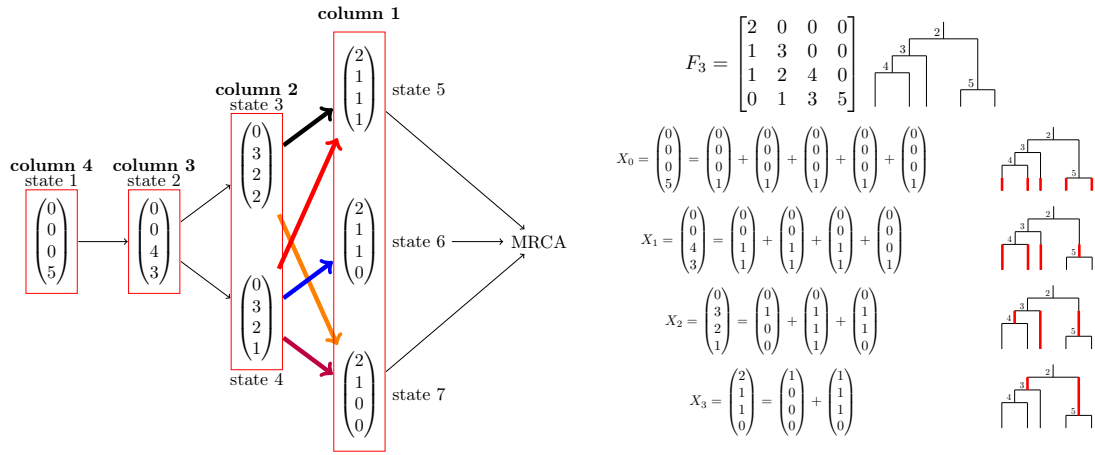
In Figure \ref{fig:twoTrees}C, we highlight the edges corresponding to the fourth column $\bx_6 = (0,0,0,5,4,4,3,2,1,0)^\transpose$. Inspired by this, we decompose $\bx_6$ in terms of its lineages as follows
\begin{align}
    \bx_6 = \begin{pmatrix}
    0 \\ 0\\ 0 \\ 5 \\4 \\4 \\ 3\\ 2 \\ 1\\ 0
    \end{pmatrix} =
    \underbrace{\begin{pmatrix}
    0 \\ 0\\ 0 \\ 1 \\ 0 \\ 0 \\ 0\\ 0 \\ 0\\ 0
    \end{pmatrix}}_{\text{a}} +
    \underbrace{\begin{pmatrix}
    0 \\ 0\\ 0 \\ 1 \\ 1 \\ 1 \\ 1\\ 1 \\ 0\\ 0
    \end{pmatrix}}_{\text{b}} +
    \underbrace{\begin{pmatrix}
    0 \\ 0\\ 0 \\ 1 \\ 1 \\ 1 \\ 1\\ 0 \\ 0\\ 0
    \end{pmatrix}}_{\text{c}} +
    \underbrace{\begin{pmatrix}
    0 \\ 0\\ 0 \\ 1 \\ 1 \\ 1 \\ 0\\ 0 \\ 0\\ 0
    \end{pmatrix}}_{\text{d}} +
    \underbrace{\begin{pmatrix}
    0 \\ 0\\ 0 \\ 1 \\ 1 \\ 1 \\ 1\\ 1 \\ 1\\ 0
    \end{pmatrix}}_{\text{e}}.
    \label{eq:X6decomp}
\end{align}
Here, each summand represents a lineage, and the entries ($0$ or $1$) indicate whether or not the lineage is present at a given layer of the tree. As only one pair of lineages can coalesce at a time, at most one lineage can terminate at each layer. This makes the decomposition unique. The decomposition allows us to derive transition probabilities: As all lineages are distinct, all $\binom{5}{2} = 10$ possible ways of selecting $2$ of the $5$ lineages to coalesce lead to distinct new states. If we assume a Kingman coalescent \cite{Coalescent}, all transitions occur with the same probability of $1/10$. In the Appendix, we derive the full transition probability matrix, also including the case of non-distinct lineages.

Suppose lineages a and b coalesce, as is the case in Figure \ref{fig:twoTrees}. Considering Figure \ref{fig:twoTrees}D, we see how the state $\bx_7$ "inherits" the lineages c,d, and e in \eqref{eq:X6decomp} with the third entry of $\bx_6$ changed from $0$ to $1$ for $\bx_7$ as these are now present also at the next layer in the tree. Furthermore, a new lineage $(0,0,1,0,\cdots,0)^\transpose$ is formed from the coalescence of a and b.
\begin{align}
    \underbrace{\begin{pmatrix}
    0 \\ 0\\ 1 \\ 0 \\ 0 \\ 0 \\ 0\\ 0 \\ 0\\ 0
    \end{pmatrix}}_{\text{a}+\text{b}}
    +
    \underbrace{\begin{pmatrix}
    0 \\ 0\\ 1 \\ 1 \\ 1 \\ 1 \\ 1\\ 0 \\ 0\\ 0
    \end{pmatrix}}_{\text{c}^\star} +
    \underbrace{\begin{pmatrix}
    0 \\ 0\\ 1 \\ 1 \\ 1 \\ 1 \\ 0\\ 0 \\ 0\\ 0
    \end{pmatrix}}_{\text{d}^\star} +
    \underbrace{\begin{pmatrix}
    0 \\ 0\\ 1 \\ 1 \\ 1 \\ 1 \\ 1\\ 1 \\ 1\\ 0
    \end{pmatrix}}_{\text{e}^\star} = 
    \begin{pmatrix}
    0 \\ 0\\ 4 \\ 3 \\ 3 \\ 3 \\ 2\\ 1 \\ 1\\ 0
    \end{pmatrix}
    = \bx_7.
    \label{eq:X7decomp}
\end{align}

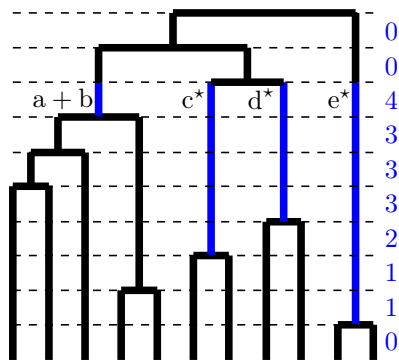
\begin{figure} [H]
  \centering
      \begin{tikzpicture}[scale=0.9]
        \node (Tree) {\begin{tikzpicture}[x=1pt,y=1pt,line width=0.6pt,scale=0.13]
    \tikzset{
      vnode/.style={circle,fill=black,inner sep=0pt,minimum size=2pt},
      rlabel/.style={font=\normalsize, inner sep=1pt, fill=white}
    }
    
    
    \draw[line width=1mm] (121,121) -- (121,625);
    \draw[line width=1mm] (225,121) -- (225,625);
    \draw[line width=1mm] (121,625) -- (225,625);
    \draw[line width=1mm] (173,625) -- (173,723);
    
    \draw[line width=1mm] (329,121) -- (329,723);
    \draw[line width=1mm] (173,723) -- (329,723);
    \draw[line width=1mm] (251,723) -- (251,826);
    
    \draw[line width=1mm] (433,121) -- (433,323);
    \draw[line width=1mm] (537,121) -- (537,323);
    \draw[line width=1mm] (433,323) -- (537,323);
    \draw[line width=1mm] (485,323) -- (485,826);
    \draw[line width=1mm] (251,826) -- (485,826);
    \draw[line width=1mm] (368,826) -- (368,1027);
    
    \draw[line width=1mm] (641,121) -- (641,424);
    \draw[line width=1mm] (743,121) -- (743,424);
    \draw[line width=1mm] (641,424) -- (743,424);
    \draw[line width=1mm] (692,424) -- (692,826);
    \draw[line width=1mm] (692,826) -- (692,927);
    
    \draw[line width=1mm] (851,121) -- (851,522);
    \draw[line width=1mm] (951,121) -- (951,522);
    \draw[line width=1mm] (851,522) -- (951,522);
    \draw[line width=1mm] (901,522) -- (901,826);
    \draw[line width=1mm] (901,826) -- (901,927);
    
    \draw[line width=1mm] (692,927) -- (901,927);
    \draw[line width=1mm] (796.5,927) -- (796.5,1027);
    
    \draw[line width=1mm] (368,1027) -- (796.5,1027);
    \draw[line width=1mm] (582.25,1027) -- (582.25,1128);
    
    \draw[line width=1mm] (1057,121) -- (1057,223);
    \draw[line width=1mm] (1159,121) -- (1159,223);
    \draw[line width=1mm] (1057,223) -- (1159,223);
    \draw[line width=1mm] (1108,223) -- (1108,826);
    \draw[line width=1mm] (1108,826) -- (1108,1128);
    
    \draw[line width=1mm] (582.25,1128) -- (1108,1128);


    
    \node[vnode,label={[rlabel]above left:11}] at (1108,223) {};
    
    \node[vnode,label={[rlabel]above left:10}] at (485,323) {};
    
    \node[vnode,label={[rlabel]above left:9}] at (692,424) {};
    
    \node[vnode,label={[rlabel]above left:8}] at (901,522) {};
    
    \node[vnode,label={[rlabel]above left:7}] at (173,625) {};
    
    \node[vnode,label={[rlabel]above left:6}] at (251,723) {};
    
    \node[vnode,label={[rlabel]above left:5}] at (368,826) {};
    
    \node[vnode,label={[rlabel]above left:4}] at (796.5,927) {};
    
    \node[vnode,label={[rlabel]above left:3}] at (582.25,1027) {};
    
    \node[vnode,label={[rlabel]above left:2}] at (845.125,1128) {};

\end{tikzpicture}};
    
        \node [right=1cm of Tree] (Fmatrix) {\begin{tikzpicture}[baseline=(m.center)]
        \node (m) {
        $\begin{bmatrix}
        2 & 0 & \color{blue}0 & \color{red}0 & 0 & 0 & 0 & 0 & 0 & 0 \\[2pt]
        1 & 3 & \color{blue}0 & \color{red}0 & 0 & 0 & 0 & 0 & 0 & 0 \\[2pt]
        1 & 2 & \color{blue}4 & \color{red}0 & 0 & 0 & 0 & 0 & 0 & 0 \\[2pt]
        1 & 1 & \color{blue}3 & \color{red}5 & 0 & 0 & 0 & 0 & 0 & 0 \\[2pt]
        1 & 1 & \color{blue}3 & \color{red}4 & 6 & 0 & 0 & 0 & 0 & 0 \\[2pt]
        1 & 1 & \color{blue}3 & \color{red}4 & 5 & 7 & 0 & 0 & 0 & 0 \\[2pt]
        1 & 1 & \color{blue}2 & \color{red}3 & 4 & 6 & 8 & 0 & 0 & 0 \\[2pt]
        1 & 1 & \color{blue}1 & \color{red}2 & 3 & 5 & 7 & 9 & 0 & 0 \\[2pt]
        1 & 1 & \color{blue}1 & \color{red}1 & 2 & 4 & 6 & 8 & 10 & 0 \\[2pt]
        0 & 0 & \color{blue}0 & \color{red}0 & 1 & 3 & 5 & 7 & 9 & 11
        \end{bmatrix}$};
    
    \foreach \i/\name in {0/\bx_9,1/\bx_8,2/\bx_7,3/\bx_6,4/\bx_5,5/\bx_4,6/\bx_3,7/\bx_2,8/\bx_1,9/\bx_0} {
        \node[above] at ($(m.north west)!{(\i+0.8)/10.6}!(m.north east)$) 
        {$\ifnum\i=2 \textcolor{blue}{\bx_7} \else \ifnum\i=3 \textcolor{red}{\bx_6} \else \name \fi \fi$};
    }
    
    \end{tikzpicture}};

    \node [below = 1cm of Tree] (column4) {\begin{tikzpicture}[x=1pt,y=1pt,line width=0.6pt,scale=0.13]
    \tikzset{
      vnode/.style={circle,fill=black,inner sep=0pt,minimum size=2pt},
      rlabel/.style={font=\normalsize, inner sep=1pt, fill=white}
    }
    
    
    
    \draw[line width=1mm] (121,121) -- (121,625);
    \draw[line width=1mm] (225,121) -- (225,625);
    \draw[line width=1mm] (121,625) -- (225,625);
    \draw[line width=1mm] (173,625) -- (173,723);
    
    \draw[line width=1mm] (329,121) -- (329,723);
    \draw[line width=1mm] (173,723) -- (329,723);
    \draw[line width=1mm,red] (251,723) -- (251,826);
    
    \draw[line width=1mm] (433,121) -- (433,323);
    \draw[line width=1mm] (537,121) -- (537,323);
    \draw[line width=1mm] (433,323) -- (537,323);
    \draw[line width=1mm,red] (485,323) -- (485,826);
    \draw[line width=1mm] (251,826) -- (485,826);
    \draw[line width=1mm] (368,826) -- (368,1027);
    
    \draw[line width=1mm] (641,121) -- (641,424);
    \draw[line width=1mm] (743,121) -- (743,424);
    \draw[line width=1mm] (641,424) -- (743,424);
    \draw[line width=1mm,red] (692,424) -- (692,826);
    \draw[line width=1mm] (692,826) -- (692,927);
    
    \draw[line width=1mm] (851,121) -- (851,522);
    \draw[line width=1mm] (951,121) -- (951,522);
    \draw[line width=1mm] (851,522) -- (951,522);
    \draw[line width=1mm,red] (901,522) -- (901,826);
    \draw[line width=1mm] (901,826) -- (901,927);
    
    \draw[line width=1mm] (692,927) -- (901,927);
    \draw[line width=1mm] (796.5,927) -- (796.5,1027);
    
    \draw[line width=1mm] (368,1027) -- (796.5,1027);
    \draw[line width=1mm] (582.25,1027) -- (582.25,1128);
    
    \draw[line width=1mm] (1057,121) -- (1057,223);
    \draw[line width=1mm] (1159,121) -- (1159,223);
    \draw[line width=1mm] (1057,223) -- (1159,223);
    \draw[line width=1mm,red] (1108,223) -- (1108,826);
    \draw[line width=1mm] (1108,826) -- (1108,1128);
    
    \draw[line width=1mm] (582.25,1128) -- (1108,1128);


    \draw[dashed] (121,223) -- (1159,223);
    \draw[dashed] (121,323) -- (1159,323);
    \draw[dashed] (121,424) -- (1159,424);
    \draw[dashed] (121,522) -- (1159,522);
    \draw[dashed] (121,625) -- (1159,625);
    \draw[dashed] (121,723) -- (1159,723);
    \draw[dashed] (121,826) -- (1159,826);
    \draw[dashed] (121,927) -- (1159,927);
    \draw[dashed] (121,1027) -- (1159,1027);
    \draw[dashed] (121,1128) -- (1159,1128);


    \node[anchor=west] at (160, 776) {\text{a}};
    \node[anchor=west] at (395, 776) {\text{b}};
    \node[anchor=west] at (605, 776) {\text{c}};
    \node[anchor=west] at (815, 776) {\text{d}};
    \node[anchor=west] at (1025, 776) {\text{e}};

    \node[inner sep=0pt,color = red,minimum size=2pt,label={[rlabel]}] at (1200,204) {0};
    \node[inner sep=0pt,color = red,minimum size=2pt,label={[rlabel]}] at (1200,304) {1};
    \node[inner sep=0pt,color = red,minimum size=2pt,label={[rlabel]}] at (1200,404) {2};
    \node[inner sep=0pt,color = red, minimum size=2pt,label={[rlabel]}] at (1200,504) {3};
    \node[inner sep=0pt,color = red,minimum size=2pt,label={[rlabel]}] at (1200,604) {4};
    \node[inner sep=0pt,color = red,minimum size=2pt,label={[rlabel]}] at (1200,704) {4};
    \node[inner sep=0pt,color = red,minimum size=2pt,label={[rlabel]}] at (1200,804) {5};
    \node[inner sep=0pt,color = red,minimum size=2pt,label={[rlabel]}] at (1200,904) {0};
    \node[inner sep=0pt,color = red,minimum size=2pt,label={[rlabel]}] at (1200,1004) {0};
    \node[inner sep=0pt,color = red,minimum size=2pt,label={[rlabel]}] at (1200,1104) {0};

\end{tikzpicture}};
    \node [right=1cm of column4] (column3) {\begin{tikzpicture}[x=1pt,y=1pt,line width=0.6pt,scale=0.13]
    \tikzset{
    vnode/.style={circle,fill=black,inner sep=0pt,minimum size=2pt},
    rlabel/.style={font=\normalsize, inner sep=1pt, fill=white}
    }

    
    
    \draw[line width=1mm] (121,121) -- (121,625);
    \draw[line width=1mm] (225,121) -- (225,625);
    \draw[line width=1mm] (121,625) -- (225,625);
    \draw[line width=1mm] (173,625) -- (173,723);
    
    \draw[line width=1mm] (329,121) -- (329,723);
    \draw[line width=1mm] (173,723) -- (329,723);
    \draw[line width=1mm] (251,723) -- (251,826);
    
    \draw[line width=1mm] (433,121) -- (433,323);
    \draw[line width=1mm] (537,121) -- (537,323);
    \draw[line width=1mm] (433,323) -- (537,323);
    \draw[line width=1mm] (485,323) -- (485,826);
    \draw[line width=1mm] (251,826) -- (485,826);
    \draw[line width=1mm,blue] (368,826) -- (368,927);
    \draw[line width=1mm] (368,927) -- (368,1027);
    
    \draw[line width=1mm] (641,121) -- (641,424);
    \draw[line width=1mm] (743,121) -- (743,424);
    \draw[line width=1mm] (641,424) -- (743,424);
    \draw[line width=1mm,blue] (692,424) -- (692,826);
    \draw[line width=1mm,blue] (692,826) -- (692,927);
    
    \draw[line width=1mm] (851,121) -- (851,522);
    \draw[line width=1mm] (951,121) -- (951,522);
    \draw[line width=1mm] (851,522) -- (951,522);
    \draw[line width=1mm,blue] (901,522) -- (901,826);
    \draw[line width=1mm,blue] (901,826) -- (901,927);
    
    \draw[line width=1mm] (692,927) -- (901,927);
    \draw[line width=1mm] (796.5,927) -- (796.5,1027);
    
    \draw[line width=1mm] (368,1027) -- (796.5,1027);
    \draw[line width=1mm] (582.25,1027) -- (582.25,1128);
    
    \draw[line width=1mm] (1057,121) -- (1057,223);
    \draw[line width=1mm] (1159,121) -- (1159,223);
    \draw[line width=1mm] (1057,223) -- (1159,223);
    \draw[line width=1mm,blue] (1108,223) -- (1108,927);
    \draw[line width=1mm] (1108,927) -- (1108,1128);
    
    \draw[line width=1mm] (582.25,1128) -- (1108,1128);
    
    \draw[dashed] (121,223) -- (1159,223);
    \draw[dashed] (121,323) -- (1159,323);
    \draw[dashed] (121,424) -- (1159,424);
    \draw[dashed] (121,522) -- (1159,522);
    \draw[dashed] (121,625) -- (1159,625);
    \draw[dashed] (121,723) -- (1159,723);
    \draw[dashed] (121,826) -- (1159,826);
    \draw[dashed] (121,927) -- (1159,927);
    \draw[dashed] (121,1027) -- (1159,1027);
    \draw[dashed] (121,1128) -- (1159,1128);

    \node at (150, 879) {$\text{a}+\text{b}$};
    \node at (580, 879) {$\text{c}^\star$};
    \node at (770, 879) {$\text{d}^\star$};
    \node at (1000, 879) {$\text{e}^\star$};

    \node[inner sep=0pt,color=blue,minimum size=2pt,label={[rlabel]}] at (1190,174) {0};
    \node[inner sep=0pt,color=blue,minimum size=2pt,label={[rlabel]}] at (1190,274) {1};
    \node[inner sep=0pt,color=blue,minimum size=2pt,label={[rlabel]}] at (1190,374) {1};
    \node[inner sep=0pt,color=blue,minimum size=2pt,label={[rlabel]}] at (1190,474) {2};
    \node[inner sep=0pt,color=blue,minimum size=2pt,label={[rlabel]}] at (1190,574) {3};
    \node[inner sep=0pt,color=blue,minimum size=2pt,label={[rlabel]}] at (1190,674) {3};
    \node[inner sep=0pt,color=blue,minimum size=2pt,label={[rlabel]}] at (1190,774) {3};
    \node[inner sep=0pt,color=blue,minimum size=2pt,label={[rlabel]}] at (1190,874) {4};
    \node[inner sep=0pt,color=blue,minimum size=2pt,label={[rlabel]}] at (1190,974) {0};
    \node[inner sep=0pt,color=blue,minimum size=2pt,label={[rlabel]}] at (1190,1074) {0};
    
    \end{tikzpicture}};

    \node (HeaderNode) at ([shift={(-1cm,3.3cm)}] Tree) {};
    \node [anchor=west] at ([shift={(-3cm,0cm)}] HeaderNode) {\large{\textbf{A:} Tree with $11$ leaves}};
    \node [anchor=west] at ([shift={(5cm,0cm)}] HeaderNode) {\large{\textbf{B:} Corresponding $\Fmat$}};
    \node [anchor=west] at ([shift={(-3cm,-6.85cm)}] HeaderNode) {\large{\textbf{C:} Interpretation of column $4$}};
    \node [anchor=west] at ([shift={(5cm,-6.85cm)}] HeaderNode) {\large{\textbf{D:} Interpretation of column $3$}};

    \end{tikzpicture}
  \caption{Coalescent interpretation of $\Fmat$ columns and transitions between them. {\bf A:} A ranked and unlabelled tree with $11$ leaves. {\bf B:} The corresponding $\Fmat$, where the third and fourth column are colored to match the interpretation below. {\bf C:} Interpretation of $\bx_6 = (0,0,0,5,4,4,3,2,1,0)^\transpose$. {\bf D:} Interpretation of $\bx_7 = (0,0,4,3,3,3,2,1,1,0)^\transpose$.}
  \label{fig:twoTrees}
\end{figure}
\noindent

This example, along with Figure \ref{fig:Figure-Fmat}D, demonstrates how transitions determined by coalescent events form $\Fmats$, suggesting that every $\Fmat$ - and thus every tree $T \in \mT_n$ - can be uniquely represented by a sequence of states, or path, through the Markov chain. We formalize this observation in the following theorem.

\begin{theorem}\label{Thm:MarkovEmbeddingTheorem}
    Every ranked unlabelled tree $T \in \mT_n$ can be uniquely encoded as a path $p = (\bx_0,\cdots,\bx_{n-2})$ through $\mX_n$.
\end{theorem}

While this provides a natural interpretation of ranked tree shapes as backward-in-time coalescent processes, we have yet to fully motivate the embedding beyond improved interpretation. One issue when dealing with the full space $\mT_n$ is its sheer size. $|\mT_n|$ equals the \textit{Euler / up-down numbers} (sequence A000111 in \cite{IntegerSeq}) \cite{CardTreeSpace} and scales asymptotically as $2\left(\frac{2}{\pi}\right)^{n+1}n!$. In contrast, $\mX_n$ grows much more slowly, as demonstrated in Figure \ref{fig:SizeComp} and formalized in the following theorem.
\begin{figure}[!bt]
    \centering
    \begin{tikzpicture}
        \node (Picture) {\includegraphics[width=0.75\linewidth]{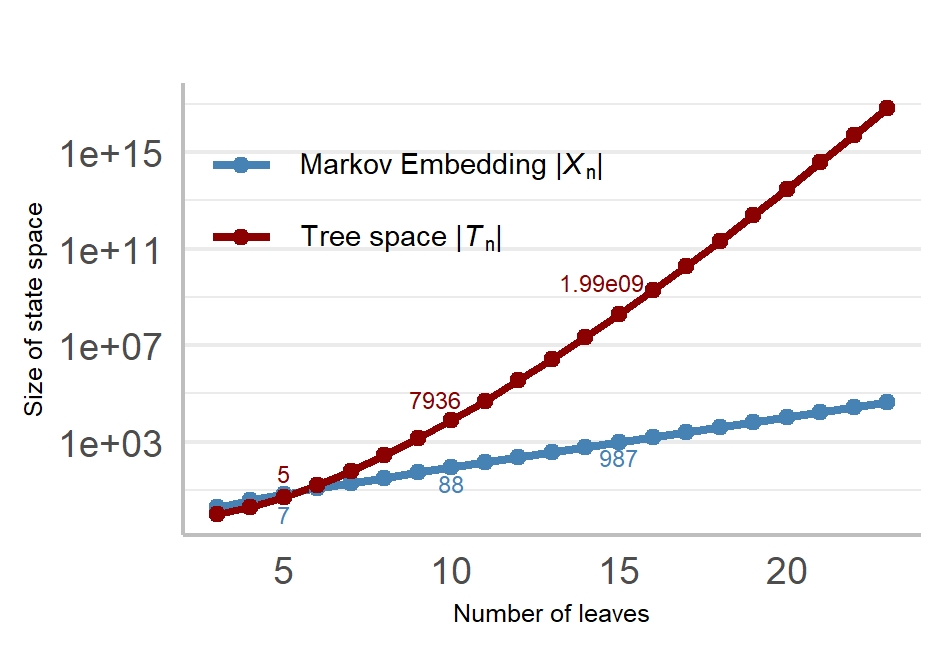}};
        \node [above=-1cm of Picture] (A) {\Large{Comparison of state space sizes}};
    \end{tikzpicture}
    \vspace{-5pt}
    \caption{The size of the state space for an increasing number of leaves. $|\mT_n|$ increases superexponentially and is many orders of magnitude larger than $|\mX_n|$, even for a moderate number of leaves. Note that $y$-axis is logarithmic.}
    \label{fig:SizeComp}
    \vspace{-10pt}
\end{figure}

\begin{theorem}\label{Thm:SizeOfSS}
    The size of $\mX_n$ equals $\text{Fib}(n+1)$, where $\text{Fib}(n)$ is the $n$th Fibonacci number. 
\end{theorem}

The proof relies on a recursive structure of $\{\mX_n\}_{n=3,4,\cdots}$ and can be found in Section $9.2$ in the appendix. This result shows that $|\mX_n|$ grows asymptotically according to $((1+\sqrt{5})/2)^{n+1}$. In the following section, we demonstrate how this reduction in state space provides substantial computational advantages.

\section{Tree summaries of distributions on $\mT_n$}
Consider a probability distribution $\mu$ on the space of ranked unlabelled trees with $n$ leaves $\mT_n$. Although summarizing real-valued probability distributions is simple via classical statistics, extending these ideas to distributions supported on spaces of discrete structures such as $\mT_n$ is substantially more challenging. Recently, \cite{SamPal} introduced a framework for summarizing such distributions via Fréchet means, defined using a distance metric induced by the $\Fmat$ representation.

Specifically, for trees $T_1,T_2\in \mT_n$, we follow \cite{FmatPaper} and consider the Euclidean norm on the corresponding $\Fmats$
\begin{align}
    d(T_1,T_2) := ||F_1-F_2||_2 = \left\{\sum_{i,j} [(F_1)_{ij}-(F_2)_{ij}]^2\right\}^{1/2}. \label{eq:DistMeasure}
\end{align}
For a distribution $\mu$ on $\mT_n$, a Fréchet mean tree is any tree that minimizes the expected squared distance to all other trees \cite{Frechet1948} and is defined through its encoding by
\begin{align}
    \bar{F} \in \argmin_{F \in \mF_n} \sum_{G \in \mF_n} ||F-G||_2^2\mu(G) = \argmin_{F \in \mF_n} ||F-M||_2^2,
\end{align}
where $M = \sum_{G \in \mF_n} G \mu(G)$, and where $\mF_n$ is the space of $(n-1) \times (n-1)$ $\Fmats$, with the latter equality following from \cite{SamPal}. While this provides a notion of a central tree, the associated optimization problem is notoriously difficult to solve to optimality.

Without the Markov embedding, the optimization must be performed directly over the huge space $\mT_n$. In practice, this problem is formulated as a mixed-integer convex program, which can be solved using general-purpose optimizers such as \texttt{Gurobi} \cite{gurobi}. However, the computation time increases rapidly with $n$; for example, for $n = 23$, the reported runtime on an Intel\textsuperscript{\textregistered} i7 processor is approximately $1000$ seconds  \cite{SamPal}.

Moreover, because $\mT_n$ is discrete, the minimizer does not need to be unique. In particular, under the ranked Kingman coalescent, we find that multiple Fréchet means exist whenever $n \geq 5$. As a consequence, measures of dispersion, and credible and interquartile sets are all dependent on which solution is returned by the optimizer.

With the Markov embedding, the optimization problem can instead be reformulated in terms of paths in the much smaller state space $\mX_n$. Due to the Markov property, the minimization of $||F-M||_2^2$ can be solved by recursion. This reformulation yields two fundamental advantages: First, the recursive structure leads to substantial computational gains. Second, all minimizing paths can be found systematically, eliminating the ambiguity induced by the non-uniqueness.

In the next section, we introduce \texttt{ViTreebi}, an algorithm that exploits the structure of $\mX_n$ to efficiently compute all Fréchet means. 
\subsection{The ViTreebi algorithm}
The \texttt{Viterbi} \cite{Viterbi, HMM} algorithm for decoding hidden Markov models computes a globally most probable path using recursion to avoid an exponential number of paths. Here we introduce the \texttt{ViTreebi} algorithm that computes \textit{all} Fréchet mean trees for a probability distribution with mean tree $M$ using a similar recursion. 

According to Theorem \ref{Thm:MarkovEmbeddingTheorem}, all trees correspond to a path through $\mX_n$. Thus, finding a Fréchet mean tree, corresponds to finding a \textit{Fréchet mean path} $\bar{p} = (\bar{\bx}_0,\cdots,\bar{\bx}_{n-2})$ through $\mX_n$, fulfilling
\begin{align}
    \bar{p} = (\bar{\bx}_0,\cdots,\bar{\bx}_{n-2}) \in \argmin_{(\bx_0, \cdots, \bx_{n-2}) : \bbp(\bx_0,\cdots,\bx_{n-2}) > 0} \sum_{t=0}^{n-2}\sum_{k = 1}^{n-1} (x_{k,t} - M_{k,n-1-t})^2.\label{eq:FréchetMeanPath}
\end{align}
To solve this minimization problem, we define for each state $\bx \in \mX_n$: its order $i(\bx)$, its tier $t(\bx)$ defined as the unique time step $t$, at which $\bx$ can occur, i.e.
\begin{align}
    t(\bx) = t \quad \text{if and only if}\quad \bbp(X_t = \bx) > 0, \label{eq:TierDef}
\end{align}
its cost, $$c(\bx) := \sum_{k=1}^{n-1} (x_k - M_{k,n-1-t(\bx)})^2,$$ and the set of paths leading to it,
\begin{align*}
    \mathcal{P}(\bx) := \{(\bx_0,\cdots,\bx_{t(\bx)-1}) : \bbp(X_0 = \bx_0,\cdots, X_{t(\bx)} = \bx)> 0\}.
\end{align*}

Using these quantities, we define a $(|\mX_n|-1) \times (n-1)$ matrix of costs $\bm{C}$
\begin{align*}
    C_{i(\bx),t} = \begin{cases}
         \min_{(\bx_0,\cdots,\bx_{t-2}) \in \mathcal{P}(\bx)} \left\{\sum_{u = 0}^{t-2} c(\bx_u)\right\} + c(\bx) & \text{if } t(\bx) = t-1\\
         \infty&\text{otherwise}.
    \end{cases}
\end{align*}
By construction,the minimization over all sample paths reduces to a minimization problem over $\bm{C}_{\cdot,n-1}$, so that
\begin{align*}
    \min_{(\bx_0, \cdots, \bx_{n-2}) : \bbp(\bx_0,\cdots,\bx_{n-2}) > 0} \sum_{t=0}^{n-2}\sum_{k = 1}^{n-1} (x_{t,k} - M_{k,n-1-t})^2 = \min_i C_{i,n-1}.
\end{align*}
Similarly to the \texttt{Viterbi} algorithm, the Fréchet mean paths can be found by recursively filling out the cost matrix followed by a backtracking procedure that uncovers the path(s) that lead to the minimum. The following theorem formalizes this procedure.
\begin{theorem}\label{ViTreebiTheorem}
    The costs can be computed recursively as $C_{1,1} = 0$ and
    \begin{align}
        C_{i(\bx),t} = \underset{i \in \mathcal{A}(\bx)}{\min} \{C_{i,t-1}\} + c(\bx),\quad t = 2,\cdots, n-1,\label{ViTreebiRecursion}
    \end{align}
    where $\mathcal{A}(\bx) := \{i(\bm{r}) : \bbp(X_t = \bx \mid X_{t-1} = \bm{r}) > 0\}$ is the set of antecedents of $\bx$.

    The Fréchet mean paths, $\bar{p} = (\bar{\bx}_0,\cdots,\bar{\bx}_{n-2})$, can be found using backtracking by letting
    \begin{align}
        i(\bar{\bx}_{n-2}) \in \underset{i}{\argmin}\{C_{i,n-1}\},\label{Backtrack1}
    \end{align}
    and 
    \begin{align}
        i(\bar{\bx}_{n-1-t}) \in \underset{i \in \mathcal{A}(\bar{\bx}_{n-t})}{\argmin}\{C_{i,n-t}\},\quad t = 2, \cdots, n-1.\label{Backtrack2}
    \end{align}
\end{theorem}
\begin{proof}
    We prove formula \eqref{ViTreebiRecursion}. The case $t = 2$ is simple, as there is only one possible transition at this step, namely the transition from state $1$ to state $2$ (see Figure \ref{fig:Figure-Fmat}C). For $\bx \in \mX_n$ with $t(\bx) > 2$ we have
    \begin{align*}
        C_{i(\bx),t} &= \min_{(\bx_0,\cdots,\bx_{t-2}) \in \mathcal{P}(\bx)} \left\{\sum_{u = 0}^{t-2} c(\bx_u)\right\} + c(\bx)\\ 
        &= \min_{\bm{r} \in \mathcal{A}(\bx)}\left\{\min_{(\bx_0,\cdots,\bx_{t-3}) \in \mathcal{P}(\bm{r})} \left\{\sum_{u=0}^{t-3} c(\bm{s}_u)\right\} +c(\bm{r})\right\} + c(\bx) \\
        & =\min_{\bm{r} \in \mathcal{A}(\bx)}\left\{C_{i(\bm{r}),t-1}\right\} + c(\bx).
    \end{align*}
\end{proof}\noindent
In Figure \ref{fig:CostMatrix6}, we demonstrate how the \texttt{ViTreebi} algorithm works in the case where $n=6$, with the underlying distribution $\mu$ corresponds to the Kingman model, which determines the costs $c(\bx)$.

\begin{figure}[H]
    \centering
    \begin{tikzpicture}[scale = 0.75]
        \node (Table) {$\renewcommand{\arraystretch}{0.85}\begin{tabular}{c|c c c c}
         $\bx$&$i(\bx)$&$t(\bx)$&$c(\bx)$&$\mathcal{A}(\bx)$  \\ \hline
         $(0,0,0,0,6)$&$1$&$0$&$0$& \\
         $(0,0,0,5,4)$&$2$&$1$&$0$&$\{1\}$ \\
         $(0,0,4,3,3)$&$3$&$2$&$9/25$&$\{2\}$\\
         $(0,0,4,3,2)$&$4$&$2$&$4/25$&$\{2\}$\\
         $(0,3,2,2,2)$&$5$&$3$&$89/100$&$\{3,4\}$\\
         $(0,3,2,2,1)$&$6$&$3$&$29/100$&$\{3,4\}$\\
         $(0,3,2,1,1)$&$7$&$3$&$29/100$&$\{3,4\}$ \\ 
         $(0,3,2,1,0)$&$8$&$3$&$169/100$&$\{4\}$ \\ 
         $(2,1,1,1,1)$&$9$&$4$&$649/900$&$\{5,6,7\}$ \\ 
         $(2,1,1,1,0)$&$10$&$4$&$469/900$&$\{6,8\}$ \\ 
         $(2,1,1,0,0)$&$11$&$4$&$469/900$&$\{7,8\}$ \\ 
         $(2,1,0,0,0)$&$12$&$4$&$769/900$&$\{5,6,7,8\}$ \\ 
    \end{tabular}$};

    \node (Cost matrix)  at ([shift={(11cm,-0.1cm)}] Table) {
        \begin{tikzpicture}[node distance=0.4cm]
    \matrix (M) [matrix of math nodes,
             left delimiter=(,
             right delimiter=),
             row sep=-2mm,
             column sep=2mm]
{
  0 & \cdot & \cdot & \cdot & \cdot \\
  \cdot & 0 & \cdot & \cdot & \cdot \\
  \cdot & \cdot & 9/25 & \cdot & \cdot \\
  \cdot & \cdot & 4/25 & \cdot & \cdot \\
  \cdot & \cdot & \cdot & 105/100 & \cdot \\
  \cdot & \cdot & \cdot & 45/100 & \cdot \\
  \cdot & \cdot & \cdot & 45/100 & \cdot \\
  \cdot & \cdot & \cdot & 185/100 & \cdot \\
  \cdot & \cdot & \cdot & \cdot & 1054/900 \\
  \cdot & \cdot & \cdot & \cdot & 874/900 \\
  \cdot & \cdot & \cdot & \cdot & 874/900 \\
  \cdot & \cdot & \cdot & \cdot & 1174/900 \\
};
    
    \node[above] at ($(M-1-1.center)+(-0.3mm,2.8mm)$) {\textbf{0}};
    \node[above] at ($(M-1-2.center)+(0,3.14mm)$) {\textbf{1}};
    \node[above] at ($(M-1-3.center)+(0,3.14mm)$) {\textbf{2}};
    \node[above] at ($(M-1-4.center)+(0,3.14mm)$) {\textbf{3}};
    \node[above] at ($(M-1-5.center)+(0.3mm,2.8mm)$) {\textbf{4}};
    \foreach \row in {1,...,12} {
    \node[right=2mm of M.east |- M-\row-1.center] {\textbf{\row}};
}
    
    \draw[->, red, thick] (M-10-5.west) to[out=180,in=0] (M-6-4.east);
    \draw[->, red, thick] (M-6-4.west) to[out=180,in=0] (M-4-3.east);
    \draw[->, red, thick] (M-4-3.west) to[out=180,in=0] (M-2-2.east);
    \draw[->, red, thick] (M-2-2.west) to[out=180,in=0] (M-1-1.east);

    \draw[->, blue, thick] (M-11-5.west) to[out=180,in=0] (M-7-4.east);
    \draw[->, blue, thick] (M-7-4.west) to[out=180,in=0] (M-4-3.east);
    \draw[->, blue, thick] (M-4-3.west) to[out=90,in=0] (M-2-2.east);
    \draw[->, blue, thick] (M-2-2.west) to[out=90,in=0] (M-1-1.east);
    \end{tikzpicture}};

    \node[anchor=west] [below=0.45cm of Table] (Tree1) {\begin{tikzpicture}
        \draw[red] (9.5, 2.85) -- (13.5,2.85);
    \draw[red] (9.5, 2.85) -- (9.5, 2.31);
    \draw[red] (8.83, 2.31) -- (10.17, 2.31);
    \draw[red] (8.83, 2.31) -- (8.83, 1.77);
    \draw[red] (8.17, 1.77) -- (9.5, 1.77);
    \draw[red] (8.17, 1.77) -- (8.17, 1.23);
    \draw[red] (7.5, 1.23) -- (8.83, 1.23);
    \draw[red] (7.5, 1.23) -- (7.5, 0.15);
    \draw[red] (8.83, 1.23) -- (8.83, 0.15);
    \draw[red] (13.5,2.85) -- (13.5,0.69);
    \draw[red] (12.83,0.69) -- (14.17,0.69);
    \draw[red] (12.83,0.69) -- (12.83,0.15);
    \draw[red] (14.17,0.69) -- (14.17,0.15);
    \draw[red] (10.17,2.31) -- (10.17, 0.15);
    \draw[red] (9.5,1.77) -- (9.5,0.15);
    \node at (11.5,3.2) {$2$};
    \node at (9.37,2.7) {$3$};
    \node at (8.7,2.1) {$4$};
    \node at (8.05,1.6) {$5$};
    \node at (13.35,1.1) {$6$};
    \end{tikzpicture}};

    \node[anchor=west] [right=of Tree1] (Tree2) {\begin{tikzpicture}
    \draw[blue] (9.5, -0.12) -- (13.5,-0.12);
    \draw[blue] (9.5, -0.12) -- (9.5, -0.66);
    \draw[blue] (8.83, -0.66) -- (10.17, -0.66);
    \draw[blue] (8.83, -0.66) -- (8.83, -1.2);
    \draw[blue] (8.17, -1.2) -- (9.5, -1.2);
    \draw[blue] (8.17, -1.2) -- (8.17, -2.28);
    \draw[blue] (7.5, -2.28) -- (8.83, -2.28);
    \draw[blue] (7.5, -2.28) -- (7.5, -2.82);
    \draw[blue] (8.83, -2.28) -- (8.83, -2.82);
    \draw[blue] (9.5, -1.2) -- (9.5,-2.82);
    \draw[blue] (10.17, -0.66) -- (10.17, -2.82);
    \draw[blue] (13.5, -0.12) -- (13.5, -1.74);
    \draw[blue] (12.83, -1.74) -- (14.17, -1.74);
    \draw[blue] (12.83, -1.74) -- (12.83,-2.82);
    \draw[blue] (14.17, -1.74) -- (14.17,-2.82);

    \node at (11.25, 0.05) {$2$};
    \node at (9.1, -0.45) {$3$};
    \node at (8.43, -1.05) {$4$};
    \node at (13.1, -1.6) {$5$};
    \node at (7.77, -2.1) {$6$};
    \end{tikzpicture}};

    \node (HeaderNode) at ([shift={(-5.2cm,3.8cm)}] Table)  {};
    \node[anchor=west] at ([shift={(0cm,0cm)}] HeaderNode) {\Large{\textbf{A:} State quantities for $n=6$}};
    \node[anchor=west] at ([shift={(10cm,0cm)}] HeaderNode) {\Large{\textbf{B:} Cost matrix with $n=6$}};
    \node[anchor=west] at ([shift={(0cm,-7.6cm)}] HeaderNode) {\Large{\textbf{C:} The two Kingman Fréchet mean trees for $n=6$}};
    \end{tikzpicture}
    \caption{{\bf A:} Table containing all non-absorbing states for $n=6$, along the quantities used in the proof of Theorem \ref{ViTreebiTheorem} {\bf B:} The cost matrix of the \texttt{ViTreebi} algorithm with $n=6$ lineages. We find that the cheapest cost is $874/900$ obtained by the two Fréchet paths $\bar{p}_1 = (1,2,4,6,10)$ (red), $\bar{p}_2 = (1,2,4,7,11)$ (blue). {\bf C:} The corresponding Fréchet mean trees drawn with colors matching the arrows.}
    \label{fig:CostMatrix6}
\end{figure}
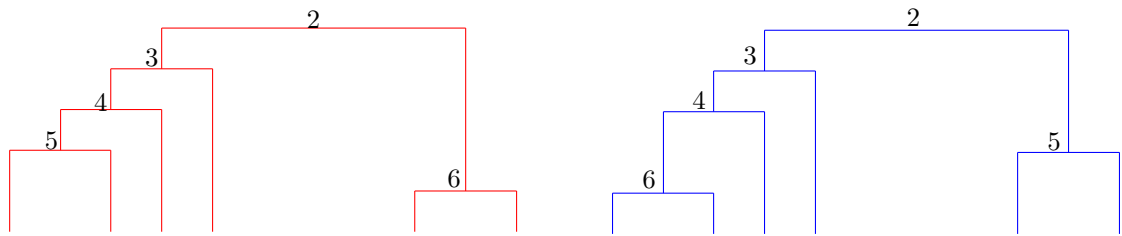

While this example illustrates the mechanics of \texttt{ViTreebi} for a small tree, the algorithm scales efficiently to larger trees. Figure \ref{fig:ViTreebiFigure} demonstrates \texttt{ViTreebi}'s ability to find all Fréchet means with dramatically reduced computation time.
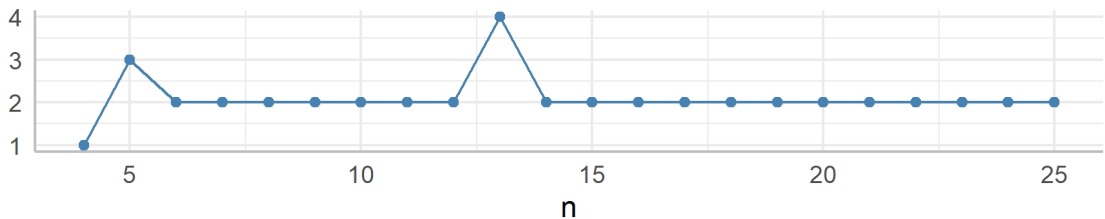
\begin{figure} [H]
    \centering
    \begin{tikzpicture}[scale = 0.95]
    \node (CompTimePlot) {\includegraphics[scale=0.85]{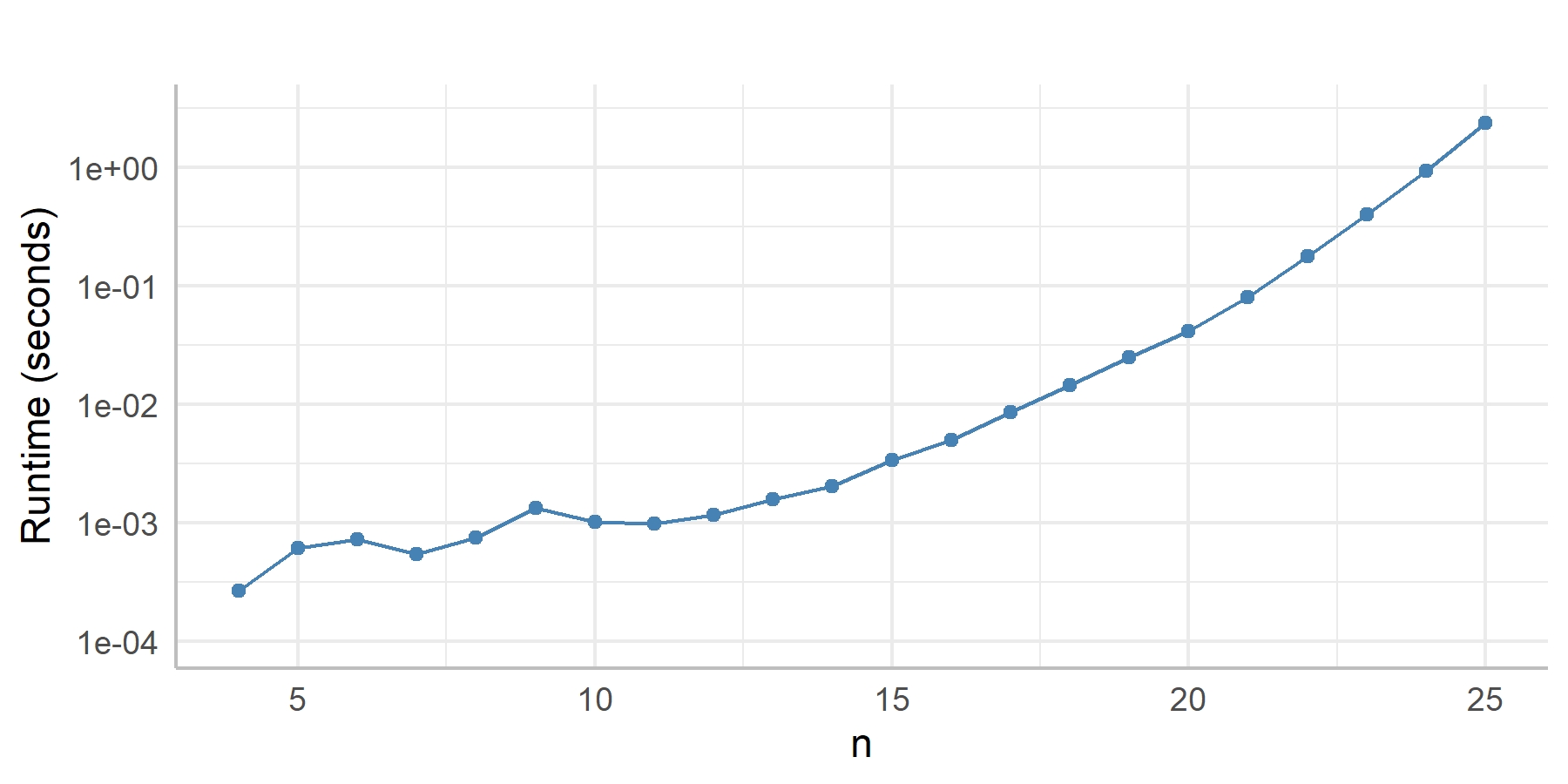}};
    \node[anchor=west] at ([shift={(-7.5cm,-5.25cm)}] CompTimePlot) {\includegraphics{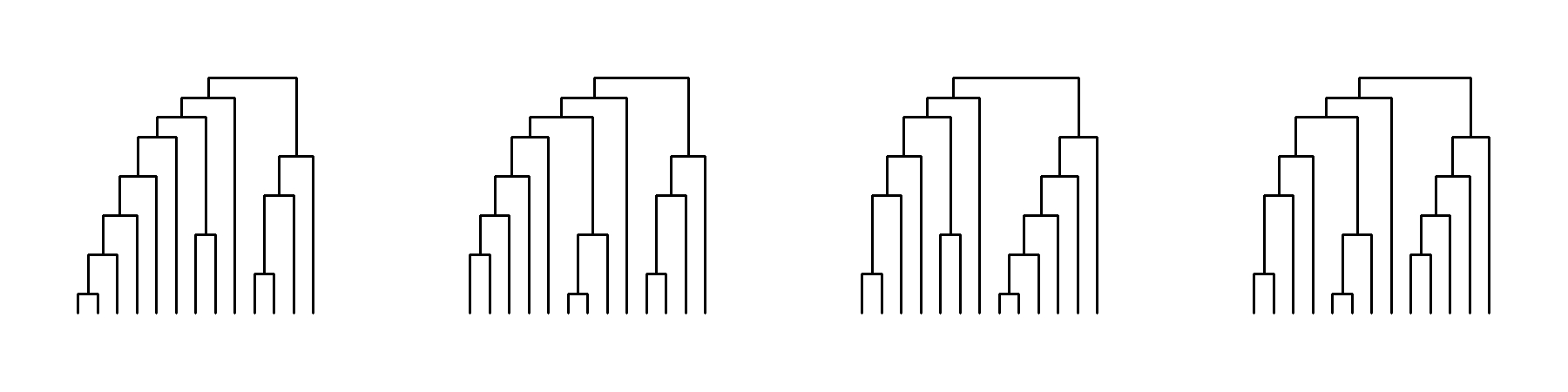}};
    \node[anchor=west] at ([shift={(-7.5cm,-8.85cm)}] CompTimePlot) {\includegraphics{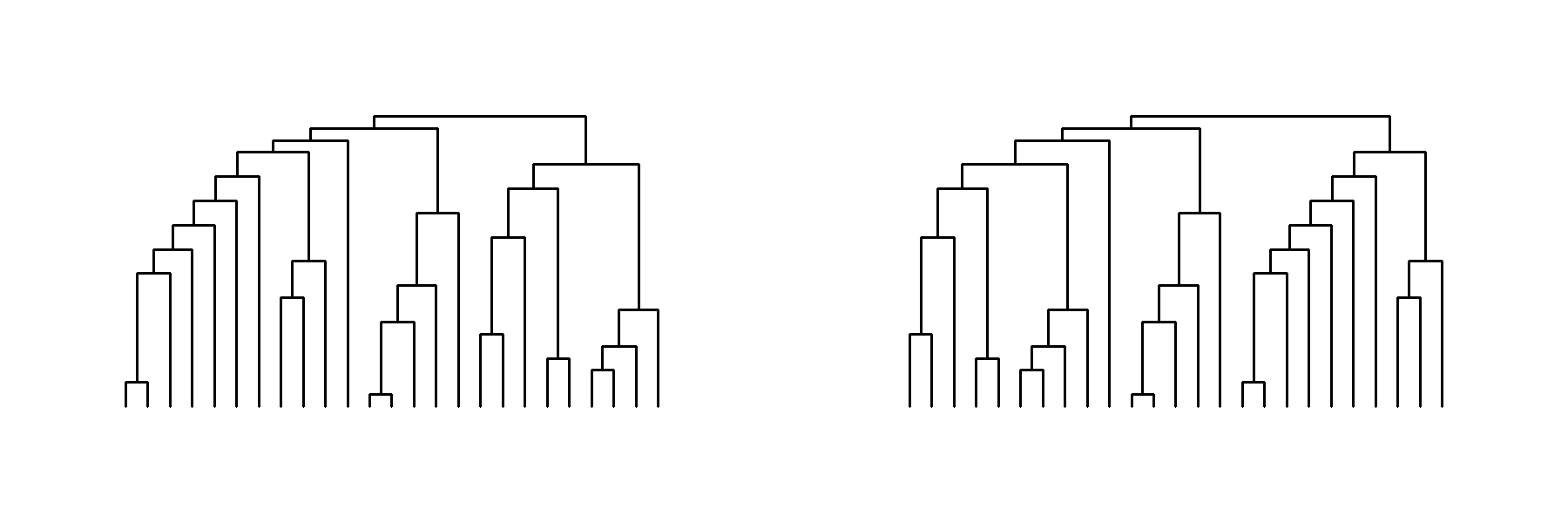}};
    \node at ([shift={(0cm,-12.75cm)}] CompTimePlot) {\includegraphics{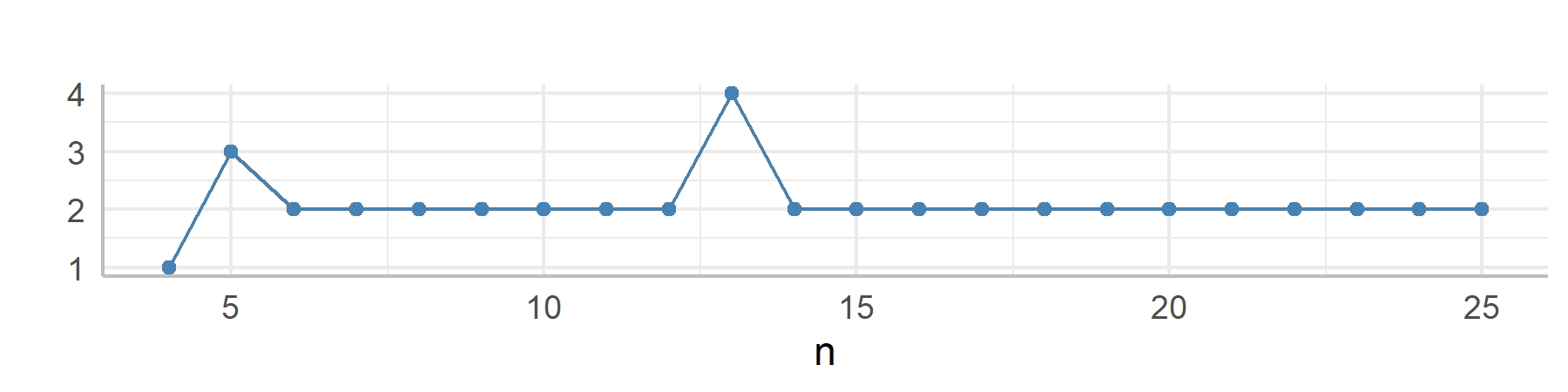}};

    \node (HeaderNode) at ([shift={(-8cm,3cm)}] CompTimePlot)  {};
    \node[anchor=west] at ([shift={(0cm,0cm)}] HeaderNode) {\Large{\textbf{A:} Runtime against number of leaves}};
    \node[anchor=west] at ([shift={(0cm,-6.5cm)}] HeaderNode) {\Large{\textbf{B:} The four Kingman Fréchet mean trees with $n=13$}};
    \node[anchor=west] at ([shift={(0cm,-9.8cm)}] HeaderNode) {\Large{\textbf{C:} The two Kingman Fréchet mean trees with $n=25$}};

    \node[anchor=west] at ([shift={(0cm,-13.9cm)}] HeaderNode) {\Large{\textbf{D:} The number of Kingman Fréchet mean trees against number of leaves}};
    \end{tikzpicture}
    \caption{Results from the \texttt{ViTreebi} algorithm. {\bf A:} average running times based on $1000$ runs for increasing number of leaves $n$. Computations were performed on a laptop with an Intel\textsuperscript{\textregistered} i5 processor. {\bf B:} we demonstrate \texttt{ViTreebi}'s ability to find all Fréchet means. {\bf C:} the two Fréchet means with $n=25$. {\bf D:} the total number of Kingman Fréchet means for each $n$ between $4$ and $25$.} 
    \label{fig:ViTreebiFigure}
\end{figure}

For example, both Fréchet means for $n=25$ lineages depicted in Figure \ref{fig:ViTreebiFigure}C are computed in less than three seconds. This represents a dramatic improvement not only in computation time, but also in the maximum tree size for which exact solutions can be obtained, compared to \cite{SamPal}. The combination of this computational speed-up and the ability to find all Fréchet means instead of just one highlights the practical importance of the Markov embedding. 

\section{Balance indices in the ranked coalescent}
The balance of phylogenetic trees has received considerable attention in the literature, as tree topology and tree balance can greatly influence phylogenetic diversity \cite{PhylogenDiversity}. The monograph \cite{TreeBalance} provides an extensive list of tree balance indices. In this section, we propose using the $\Fmat$ representation to construct two indices of tree balance, and in the following section we provide matrix analytical formulae to describe the distribution of these indices.

Recall the decomposition \eqref{eq:X6decomp}. Lineages that have yet to coalesce and are external in the tree take the shape $(0,\cdots,0,1,\cdots,1)^\transpose$. In balanced trees, these external lineages tend to coalesce early, leaving few long-lived external lineages. In contrast, imbalanced trees retain many such long-lived external lineages. We therefore propose using the total number of external lineages across all coalescent events as an index of tree balance and name this the external branch length. For a ranked unlabelled tree with $\Fmat$ $F$, this is given by
\begin{align}
    E(F) := \sum_{j = 1}^{n-1} F_{n-1,j}.\label{eq:E}
\end{align}

Although $E$ measures tree balance through the number of external lineages alone, it does not account for differences between internal lineages. To further emphasize imbalance, we seek an index that accounts not only for the presence of external lineages, but also for how long lineages, internal or external, persist through successive coalescent events.
\begin{figure}[H]
    \centering
    \begin{tikzpicture}
        \node (Tree1) {\begin{tikzpicture}[line width = 0.5mm]
        \node at (-5, 0.4) {\Large{\textbf{A}}};
        \draw (-1, 0.4) -- (-1, 0);
        \draw (-2, 0) -- (0,0);
        \draw [red](0, 0) -- (0, -3);
        \draw (-2, 0) -- (-2, -0.33);
        \draw (-2.5, -0.33) -- (-1.5, -0.33);
        \draw (-1.5, -0.33) -- (-1.5, -1.67);
        \draw (-2.5, -0.33) -- (-2.5, -0.67);
        \draw (-2.83, -0.67) -- (-2.17,-0.67);
        \draw [red](-2.17, -0.67) -- (-2.17, -3);
        \draw (-2.83, -0.67) -- (-2.83, -1);
        \draw (-3.16, -1) -- (-2.5, -1);
        \draw [red](-2.5, -1) -- (-2.5, -3);
        \draw (-3.16, -1) -- (-3.16, -1.33);
        \draw (-3.5, -1.33) -- (-2.83, -1.33);
        \draw (-1.83, -1.67) -- (-1.17,-1.67);
        \draw [red](-1.83, -1.67) -- (-1.83,-3);
        \draw [red](-1.17, -1.67) -- (-1.17,-3);
        \draw (-3.5, -1.33) -- (-3.5, -2);
        \draw [red](-2.83, -1.33) -- (-2.83, -3);
        \draw (-3.83, -2) -- (-3.17,-2);
        \draw (-3.83,-2) -- (-3.83,-2.33);
        \draw [red](-3.17,-2) -- (-3.17, -3);
        \draw (-4.16,-2.33) -- (-3.5, -2.33);
        \draw (-4.16,-2.33) -- (-4.16,-2.67);
        \draw [red](-3.5,-2.33) -- (-3.5,-3);
        \draw (-4.5,-2.67) -- (-3.83, -2.67);
        \draw [red](-4.5,-2.67) -- (-4.5,-3);
        \draw [red](-3.83,-2.67) -- (-3.83,-3);

        \node[text width = 1cm] at(-0.8,0.2) {$2$};
        \node[text width = 1cm] at(-1.75,-0.17) {$3$};
        \node[text width = 1cm] at(-2.25,-0.5) {$4$};
        \node[text width = 1cm] at(-2.67,-0.83) {$5$};
        \node[text width = 1cm] at(-3,-1.17) {$6$};
        \node[text width = 1cm] at(-1.25,-1.5) {$7$};
        \node[text width = 1cm] at(-3.33,-1.83) {$8$};
        \node[text width = 1cm] at(-3.67,-2.17) {$9$};
        \node[text width = 1cm] at(-4.17,-2.5) {$10$};

        \draw[gray, dashed] (-4.5, -2.67) -- (0.5, -2.67);
        \draw[gray, dashed] (-4.5, -2.33) -- (0.5, -2.33);
        \draw[gray, dashed] (-4.5, -2) -- (0.5, -2);
        \draw[gray, dashed] (-4.5, -1.67) -- (0.5, -1.67);
        \draw[gray, dashed] (-4.5, -1.33) -- (0.5, -1.33);
        \draw[gray, dashed] (-4.5, -1) -- (0.5, -1);
        \draw[gray, dashed] (-4.5, -0.67) -- (0.5, -0.67);
        \draw[gray, dashed] (-4.5, -0.33) -- (0.5, -0.33);

        \node[red, text width = 0.8cm] at(0.5,-2.83) {$10$};
        \node[red, text width = 0.8cm] at(0.5,-2.5) {$8$};
        \node[red, text width = 0.8cm] at(0.5,-2.16) {$7$};
        \node[red, text width = 0.8cm] at(0.5,-1.83) {$6$};
        \node[red, text width = 0.8cm] at(0.5,-1.5) {$4$};
        \node[red, text width = 0.8cm] at(0.5,-1.16) {$3$};
        \node[red, text width = 0.8cm] at(0.5,-0.83) {$2$};
        \node[red, text width = 0.8cm] at(0.5,-0.5) {$1$};
        \node[red, text width = 0.8cm] at(0.5,-0.16) {$1$};
    \end{tikzpicture}};

    \node [right=of Tree1] (Tree 2) {\begin{tikzpicture}[line width = 0.5mm]
        \node at (-2.83, 0.4) {\Large{\textbf{B}}};
        \draw(0.33, 0.4) -- (0.33, 0);
        \draw(-1, 0) -- (1.67, 0);
        \draw(-1, 0) -- (-1, -0.33);
        \draw(-1.67,-0.33) -- (0,-0.33);
        \draw(-1.67,-0.33) -- (-1.67,-1);
        \draw(1.67,0) -- (1.67,-0.67);
        \draw(1.33, -0.67) -- (2, -0.67);
        \draw(0, -0.33) -- (0, -2);
        \draw(-2, -1) -- (-1, -1);
        \draw(-2, -1) -- (-2,-2.33);
        \draw(-1,-1) -- (-1,-2.67);
        \draw (-2.33,-2.33) -- (-1.67,-2.33);
        \draw [red](-2.33,-2.33) -- (-2.33,-3);
        \draw [red](-1.67,-2.33) -- (-1.67, -3);
        \draw (-1.33,-2.67) -- (-0.67, -2.67);
        \draw [red](-1.33,-2.67) -- (-1.33,-3);
        \draw [red](-0.67,-2.67) -- (-0.67,-3);
        \draw (-0.33,-2) -- (0.33,-2);
        \draw [red](-0.33,-2) -- (-0.33,-3);
        \draw [red](0.33,-2) -- (0.33,-3);
        \draw (1.33, -0.67) -- (1.33,-1.33);
        \draw (1,-1.33) -- (1.67,-1.33);
        \draw (1,-1.33) -- (1,-1.67);
        \draw [red](1.67,-1.33) -- (1.67,-3);
        \draw (0.67,-1.67) -- (1.33,-1.67);
        \draw [red](0.67,-1.67) -- (0.67,-3);
        \draw [red](1.33,-1.67) -- (1.33,-3);
        \draw [red](2,-0.67) -- (2,-3);
        
        \node[text width = 1cm] at(-0.1,0.2) {$2$};
        \node[text width = 1cm] at(-1.4,-0.17) {$3$};
        \node[text width = 1cm] at(1.3,-0.5) {$4$};
        \node[text width = 1cm] at(-2.1,-0.83) {$5$};
        \node[text width = 1cm] at(0.97,-1.17) {$6$};
        \node[text width = 1cm] at(0.64,-1.5) {$7$};
        \node[text width = 1cm] at(-0.4,-1.83) {$8$};
        \node[text width = 1cm] at(-2.4,-2.17) {$9$};
        \node[text width = 1cm] at(-1.5,-2.5) {$10$};
        
        \draw[gray, dashed] (-2.25, -2.67) -- (2.5, -2.67);
        \draw[gray, dashed] (-2.25, -2.33) -- (2.5, -2.33);
        \draw[gray, dashed] (-2.25, -2) -- (2.5, -2);
        \draw[gray, dashed] (-2.25, -1.67) -- (2.5, -1.67);
        \draw[gray, dashed] (-2.25, -1.33) -- (2.5, -1.33);
        \draw[gray, dashed] (-2.25, -1) -- (2.5, -1);
        \draw[gray, dashed] (-2.25, -0.67) -- (2.5, -0.67);
        \draw[gray, dashed] (-2.25, -0.33) -- (2.5, -0.33);
        
        \node[red, text width = 0.8cm] at(2,-2.83) {$10$};
        \node[red, text width = 0.8cm] at(2,-2.5) {$8$};
        \node[red, text width = 0.8cm] at(2,-2.16) {$6$};
        \node[red, text width = 0.8cm] at(2,-1.83) {$4$};
        \node[red, text width = 0.8cm] at(2,-1.5) {$2$};
        \node[red, text width = 0.8cm] at(2,-1.16) {$1$};
        \node[red, text width = 0.8cm] at(2,-0.83) {$1$};
        \node[red, text width = 0.8cm] at(2,-0.5) {$0$};
        \node[red, text width = 0.8cm] at(2,-0.16) {$0$};

    \end{tikzpicture}};

    \node [below =0cm of Tree1] (F1) {$F_1 =$\begin{tikzpicture}[baseline=(m.center)]
\node (m) [matrix of math nodes,
            left delimiter=(,
            right delimiter=),
            nodes in empty cells,
            nodes={minimum size=4.5mm, anchor=center},
            row sep = 0mm,
            column sep = 0mm] {
    2 & 0 & 0 & 0 & 0 & 0 & 0 & 0 & 0 \\
    1 & 3 & 0 & 0 & 0 & 0 & 0 & 0 & 0 \\
    1 & 2 & 4 & 0 & 0 & 0 & 0 & 0 & 0 \\
    1 & 2 & 3 & 5 & 0 & 0 & 0 & 0 & 0 \\
    1 & 2 & 3 & 4 & 6 & 0 & 0 & 0 & 0 \\
    1 & 1 & 2 & 3 & 5 & 7 & 0 & 0 & 0 \\
    1 & 1 & 2 & 3 & 4 & 6 & 8 & 0 & 0 \\
    1 & 1 & 2 & 3 & 4 & 6 & 7 & 9 & 0 \\
    1 & 1 & 2 & 3 & 4 & 6 & 7 & 8 & 10 \\
};

\foreach \j in {1,...,9} {
    \draw[red] (m-9-\j.north west) rectangle (m-9-\j.south east);
}

\foreach \i in {3,...,9} {
    \foreach \j in {1,...,\numexpr\i-2\relax} {
        \draw[blue] (m-\i-\j.north west) rectangle (m-\i-\j.south east);
    }
}
\foreach \j in {1,...,7} {
    \draw[purple, ultra thick] (m-9-\j.north west) rectangle (m-9-\j.south east);
}
\end{tikzpicture}};
\node [right=of F1] (F2) {    $F_2 =$ \begin{tikzpicture}[baseline=(m.center)]
\node (m) [matrix of math nodes,
            left delimiter=(,
            right delimiter=),
            nodes in empty cells,
            nodes={minimum size=4.5mm, anchor=center},
            row sep = 0mm,
            column sep = 0mm] {
    2 & 0 & 0 & 0 & 0 & 0 & 0 & 0 & 0 \\
    1 & 3 & 0 & 0 & 0 & 0 & 0 & 0 & 0 \\
    0 & 2 & 4 & 0 & 0 & 0 & 0 & 0 & 0 \\
    0 & 1 & 3 & 5 & 0 & 0 & 0 & 0 & 0 \\
    0 & 1 & 2 & 4 & 6 & 0 & 0 & 0 & 0 \\
    0 & 1 & 2 & 4 & 5 & 7 & 0 & 0 & 0 \\
    0 & 0 & 1 & 3 & 3 & 6 & 8 & 0 & 0 \\
    0 & 0 & 1 & 2 & 3 & 5 & 7 & 9 & 0 \\
    0 & 0 & 1 & 1 & 2 & 4 & 6 & 8 & 10 \\
};

\foreach \j in {1,...,9} {
    \draw[red] (m-9-\j.north west) rectangle (m-9-\j.south east);
}

\foreach \i in {3,...,9} {
    \foreach \j in {1,...,\numexpr\i-2\relax} {
        \draw[blue] (m-\i-\j.north west) rectangle (m-\i-\j.south east);
    }
}
\foreach \j in {1,...,7} {
    \draw[purple, ultra thick] (m-9-\j.north west) rectangle (m-9-\j.south east);
}
\end{tikzpicture}};
\node[above=0cm of $(Tree1.north)!0.5!(Tree 2.north)$]
{\Large{Illustration of balance indices $E$ and $S$}};
    \end{tikzpicture}
    \caption{Demonstrating that $E$ and $S$ capture tree balance. {\bf A:} An imbalanced tree and its corresponding $\Fmat$. {\bf B:} A balanced tree along with its corresponding $\Fmat$. For the left tree we find $E(F_1) = 42$ and $S(F_1) = 69$. For the right tree $E(F_2) = 32$ and $S(F_2) = 43$. Entries contributing to $E$ are highlighted in red, entries contributing to $S$ are highlighted in blue, and entries contributing to both are highlighted in purple.
}
    \label{fig:EFig}
\end{figure}
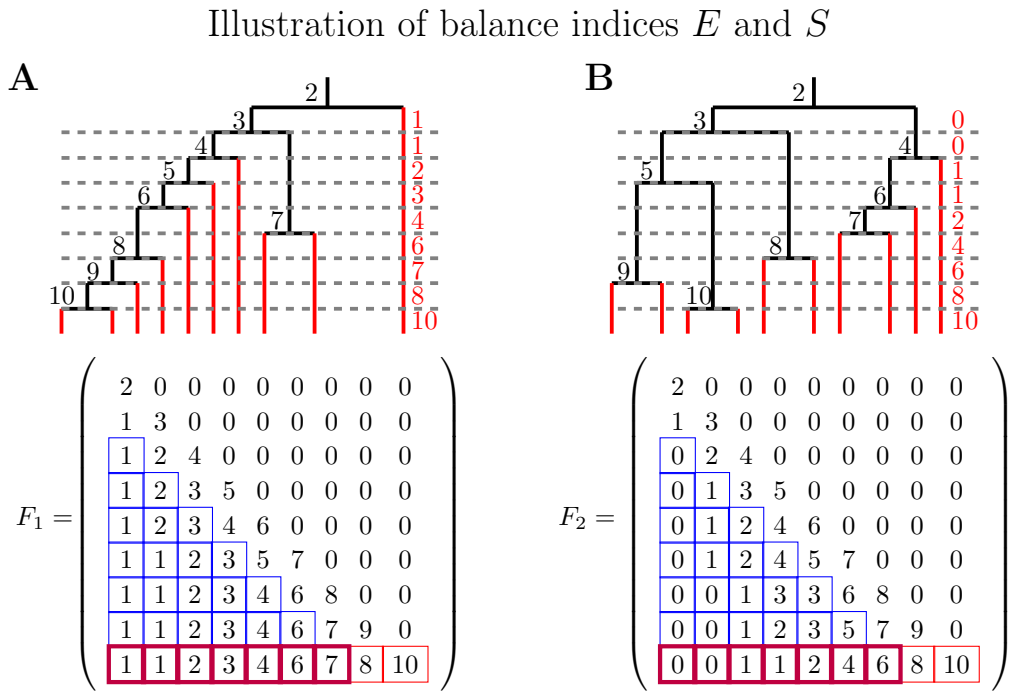\noindent 
As an example, recall the decomposition \eqref{eq:X6decomp} and consider two internal lineages given by $(0,\cdots,1,0,\cdots,0,0)^\transpose$ and $(0,\cdots,1,1,\cdots,1,0)^\transpose$, respectively. While the lineage $(0,\cdots,1,1,\cdots,1,0)^\transpose$ has only coalesced once, the lineage $(0,\cdots,1,0,\cdots,0,0)^\transpose$ is rooted much deeper in the tree. Such persistence is encoded in the non-fixed entries of the corresponding $\Fmat$. Each non fixed entry $F_{ij}, j = 1,\cdots,n-3, i = j+2,\cdots,n-1$ counts the number of lineages extant at time $j$ that survive until time $i$. Thus, summing these entries measures the total lineage persistence across the tree. We therefore propose the index
\begin{align}
    S(F) := \sum_{j=1}^{n-3}\sum_{i=j+2}^{n-1} F_{ij}, \label{eq:SumNonFixed}
\end{align}
which measures the total persistence of lineages across the coalescent process. 

The new balance index $S$ provides a complementary perspective to classical tree balance indices such as the Sackin and Colless indices \cite{TreeBalance}. The Sackin index measures tree balance in terms of cumulative distance of all leaves to the root, while the Colless index quantifies asymmetry at each branching event. Both indices are defined for unlabelled trees and therefore do not incorporate the ranking of coalescent events. In contrast, $S$ characterizes balance through the persistence of lineages throughout the tree and is tailored to ranked tree shapes. Despite this difference, the indices correlate well, as illustrated in Figure \ref{fig:SackinColless}.
\captionsetup{skip=0pt}
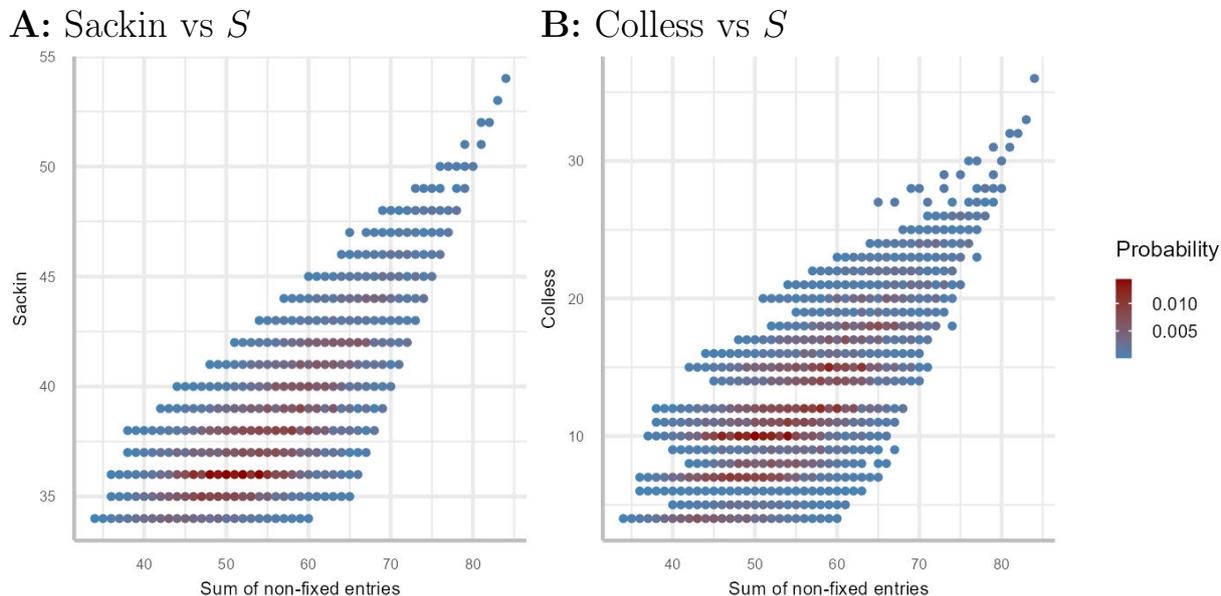
\begin{figure} [H]
    \centering
    \begin{tikzpicture}
        \node (Picture) {\includegraphics[width=1\linewidth]{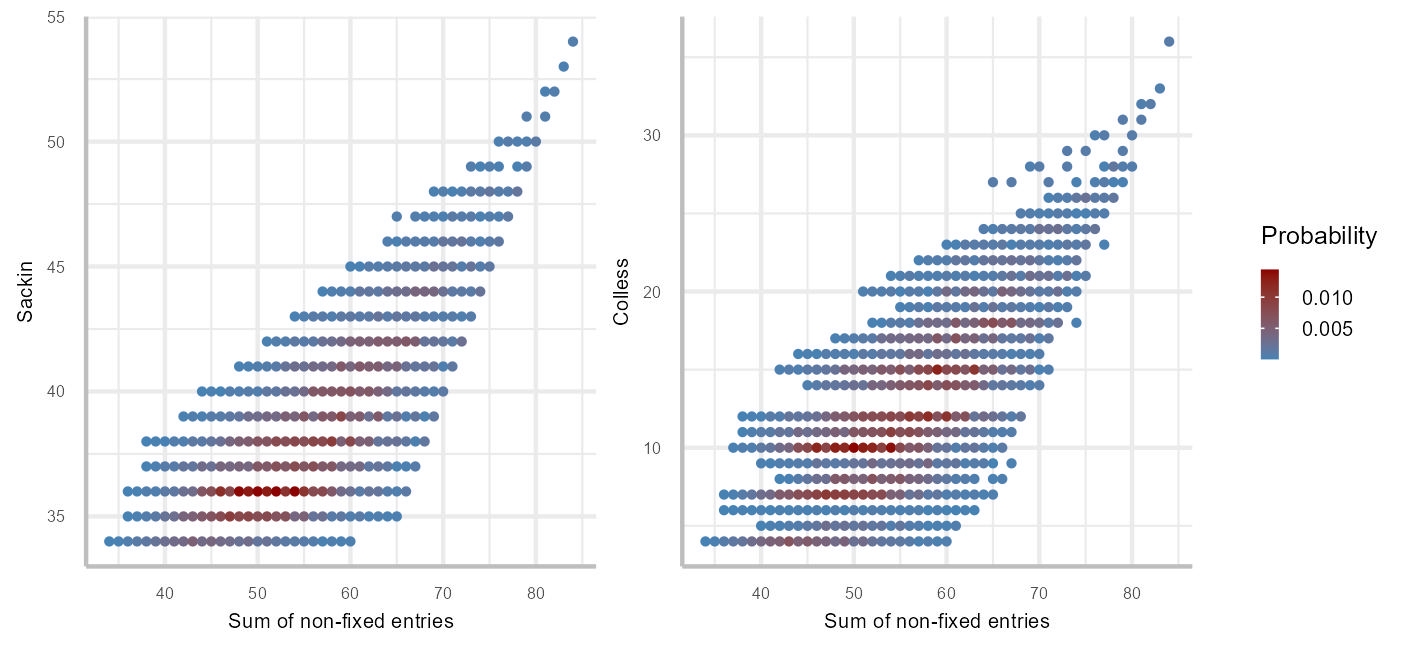}};
        \node (HeaderNode) [above=0cm of Picture] {};
        \node[anchor=west] at ([shift={(-8.25cm,0cm)}] HeaderNode) {\Large{\textbf{A:} Sackin vs $S$}};
        \node[anchor=west] at ([shift={(-1.25cm,0cm)}] HeaderNode) {\Large{\textbf{B:} Colless vs $S$}};
    \end{tikzpicture}
    
    \caption{Sackin {\bf A} and Colless {\bf B} index against $S$ for all $7936$ ranked unlabelled trees with $n=10$ leaves. Each dot is colored according to the total probability of trees corresponding to the given $(S,\text{index})$ pair under the Kingman model.}
    \label{fig:SackinColless}
\end{figure}

Compared to Sackin and Colless, $S$ spans a larger range of values, making it more sensitive to differences in ranked tree shapes. While all indices manage to capture the caterpillar as the most imbalanced (top right corner in Figure \ref{fig:SackinColless}A and B), $S$ also uniquely identifies the most balanced ranked tree, $T_{\text{bal}}$ as defined in \cite{SamPal}, by assigning it a unique minimal value. Both the Sackin and Colless index lack this ability whenever $n \neq 2^k$. Finally, the definition of $S$ allows its full distribution to be computed under any time homogeneous and bifurcating coalescent model via phase-type methods. This makes $S$ an ideal balance index for the ranked coalescent.

In Figure \ref{fig:BalanceSummaries}, we compute the balance indices for $m=1000$ trees with $n=25$ leaves simulated from the Blum-Fran\c{c}ois model \cite{Blum-Francois} using code from the \texttt{fmatrix} package from \cite{SamPal}, with moderate balance and imbalance ($\beta = 1$ and $\beta = -0.5$, respectively). The figure highlights how the proposed indices capture tree balance and motivates the tests for neutrality developed in Section 7. The theoretical framework enabling these tests is built on phase-type theory and is introduced in the following section. 
\begin{figure} [H]
    \centering
    \includegraphics[width=0.9\linewidth]{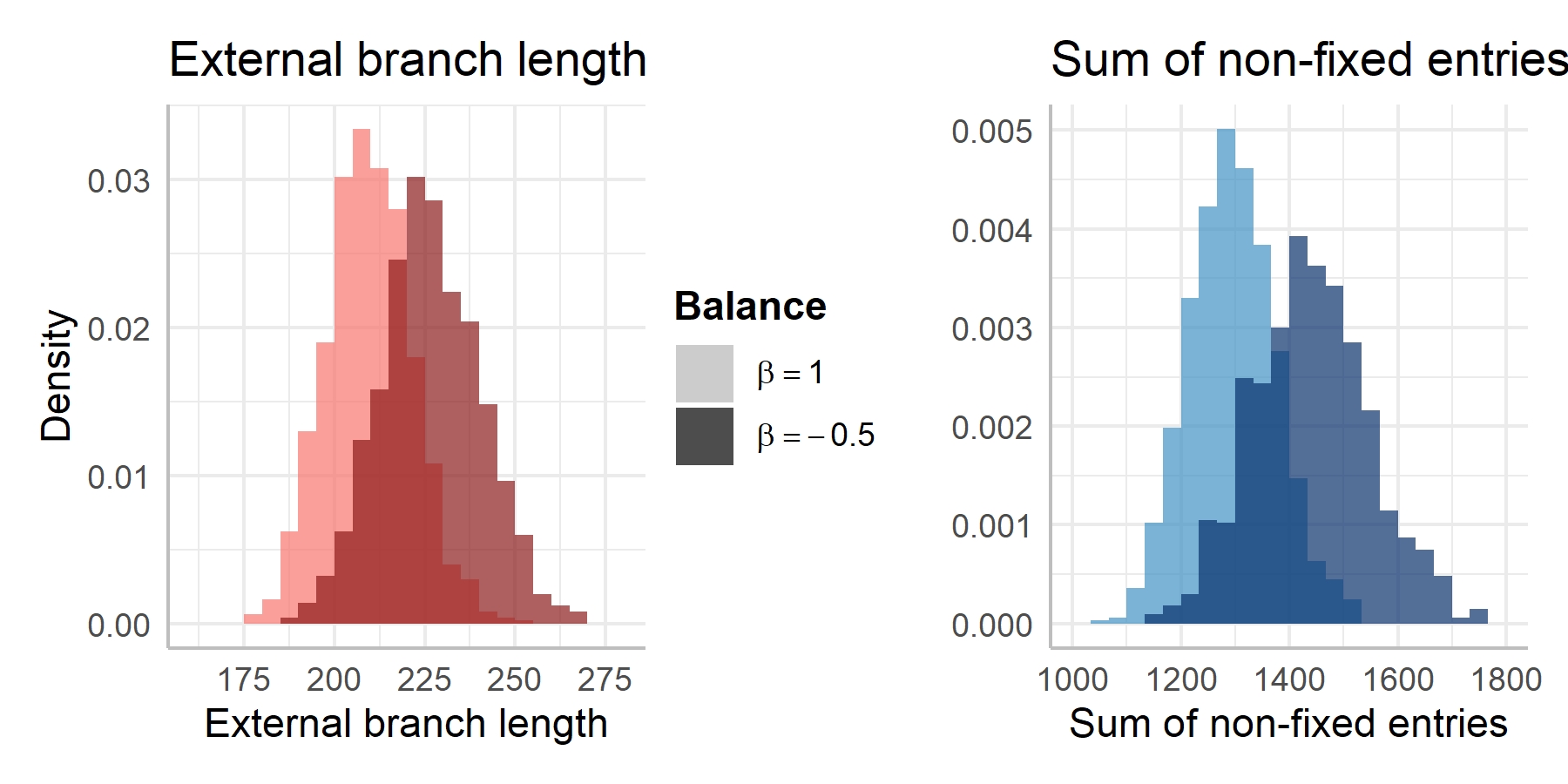}
    \caption{External branch length (left) and the sum of non-fixed entries (right) for $m = 1000$ trees with $n = 25$ leaves sampled from the Blum-Fran\c{c}ois model. Shading represents tree balance, with dark colors indicating more imbalanced trees ($\beta = 0.5$) and light colors indicating more balanced trees ($\beta = 1$).}
    \label{fig:BalanceSummaries}
\end{figure}

\section{Phase-type analysis of the ranked coalescent}
The Markov embedding provides a compact representation of genealogical topologies, capturing the order of coalescence events without labeling individual lineages. In this section, we demonstrate that the embedding enables matrix analytical expressions for distributional properties of functionals of each state, such as external branch length and sum of non-fixed entries via discrete phase-type analysis. This approach generalizes Theorem 2 in \cite{SamPal}, extending  beyond the Kingman model and enabling matrix analytic expressions for higher order moments.

Both the external branch length and the sum of non-fixed entries can be obtained by assigning certain rewards to each state in the ranked coalescent. Phase-type theory offers a strong theoretical framework for working with such reward-based quantities in absorbing Markov chains. In this section, we briefly review results from discrete phase-type theory and demonstrate their applicability to the ranked coalescent by deriving distributional quantities of the balance indices $E, S$. We refer to \cite{AzucenaPhD, MatExpo} for further information on discrete phase-type distributions.

\subsection{The discrete phase-type distribution}
Consider a time-homogeneous discrete Markov chain $\{X_t\}_{t\in\bbn_0}$ with state space $\{1,\cdots,p,p+1\}$, where states $1,\cdots,p$ are transient and the last state, $p+1$, is absorbing. We write the transition probability matrix $\bm{P}$ as
\begin{align*}
    \bm{P} = \begin{pmatrix}
        \bm{T}&\bm{t}\\0&1
    \end{pmatrix},
\end{align*}
where $\bm{T}$ is an $p \times p$ \textit{sub-transition probability matrix} and $\bm{t}=\bigl(\bm{I}-\bm{T}\bigr)\bm{e}$ is the vector of \textit{exit probabilities}. Here, $\bm{e}$ denotes the column vector of ones. We define $\tau$ as the time to absorption
\begin{align*}
    \tau := \min\{t \geq 0: X_t = p + 1\}.
\end{align*}
Denoting the initial distribution by $\bm{\pi} = (\pi_1,\cdots,\pi_p)$ such that $\bbp(X_0 = i) = \pi_i$ we say that $\tau$ has a discrete phase type distribution of order $p$ with initial distribution $\bm{\pi}$ and sub-transition probability matrix $\bm{T}$, and we write
\begin{align*}
    \tau \sim \text{DPH}_p(\bm{\pi},\bm{T}).
\end{align*}
The event $\{\tau = m\}$ occurs when the chain moves between the transient states $m-1$ times before exiting at time $m$ yielding
\begin{align}
    \bbp(\tau = m) = \bm{\pi}\bm{T}^{m-1}\bm{t}, \quad m \in \bbn. \label{eq:PMFtau}
\end{align}
We let
\begin{align}
    \bm{U} = \{u_{ij}\} = \bigl(\bm{I} - \bm{T}\bigr)^{-1} \label{eq:DefGreenMat}
\end{align}
denote the \textit{fundamental matrix}, whose entries $u_{ij}$ equal the expected number of visits to state $j$ when starting in $i$ \cite{MatExpo}. Using this notation, the moments of $\tau$ are provided in Proposition 2.7 in \cite{AzucenaPhD}. In particular,
\begin{align}
    &\bE[\tau] = \bm{\pi}\bm{U}\bm{e},\quad \bE[\tau(\tau-1)] = 2\bm{\pi}\bm{T}\bm{U}^2\bm{e},\quad \var(\tau) = \bE[\tau(\tau-1)] + \bE[\tau]\bigl(1-\bE[\tau]\bigr).\label{eq:momentstau}
\end{align}
\subsubsection{Reward transformations}
Let $\tau \sim \text{DPH}_p(\bm{\pi},\bm{T})$ and let $\{X_t\}_{t\in \bbn_0}$ be the underlying Markov chain. Define
\begin{align*}
    Y := \sum_{t = 0}^{\tau-1} r(X_t),
\end{align*}
where $r : \{1,\cdots,p\} \to \bbn_0$ is a non-negative and integer \textit{reward transform}. Theorem 5.2 in \cite{AzucenaPhD} shows by way of construction that there exist an initial distribution $\bm{\pi}^*$ and a sub-transition probability matrix $\bm{T}^*$ such that 
\begin{align*}
    Y \sim \text{DPH}_{p^*}(\bm{\pi}^*, \bm{T}^*),
\end{align*}
where $p^* = \sum_{j=1}^p r(j)$ is the size of the state space of the reward-transformed variable.

Although the moments of $Y$ can be computed using the DPH representation along with \eqref{eq:momentstau}, in practice it is more efficient to use
\begin{align}
    \bE[Y] = \bm{\pi}\bm{U}\bm{\Delta}(\bm{r})\bm{e},\quad \bE[Y^2] = \bm{\pi}\Bigl\{2(\bm{U}\bm{\Delta}(\bm{r})\}^2-\bm{U}\bm{\Delta}(\bm{r})^2\Bigr), \quad \var(Y) = \bE[Y^2]-(\bE[Y])^2\label{eq:MeanY},
\end{align}
where $\bm{\Delta}(\bm{r}) = \text{diag}(r(1),\cdots,r(p))$. 
The computations and examples presented below are implemented using the \texttt{PhaseTypeR} package \cite{PhaseTypeR}, which provides a powerful framework for working with phase-type distributions.
\subsubsection*{Example: Sum of non-fixed entries and external branch length for $n=5$}
Recallling Figure \ref{fig:Figure-Fmat}C, we let $n = 5$ and consider the ranked Kingman coalescent $\{X_t\}_{t=0,\cdots,4}$ with parameters
\begin{align}
    \bm{\pi} = (1,0,0,0,0,0,0), \quad \bm{T} = \begin{pmatrix}
        0&1&0&0&0&0&0\\
        0&0&1/2&1/2&0&0&0\\
        0&0&0&0&2/3&0&1/3\\
        0&0&0&0&1/3&1/3&1/3\\
        0&0&0&0&0&0&0\\
        0&0&0&0&0&0&0\\
        0&0&0&0&0&0&0
    \end{pmatrix}.\label{eq:paramKingmanRcoal5}
\end{align}
Define $\tau = \min\{t \geq 0 : X_t = \text{MRCA}\}$, and let $r_S :\{1,\cdots,7\} \to \bbn$ denote the reward transform that for each state $\bx$ returns the sum of the non-fixed entries. Using the notation $i(\bx)$ and $t(\bx)$ introduced in Section 3, we get
\begin{align}
    r_S(i(\bx)) = \sum_{j = n + 1 - t(\bx)}^{n-1} x_j. \label{eq:RewardTransformS}
\end{align}
In the case $n = 5$ we get 
\begin{align}
    (r_S(1),\cdots, r_S(7)) = (0,0,2,1,2,1,0). \label{eq:rS5}
\end{align}
Consequently,
\begin{align*}
    S = \sum_{t = 0}^{\tau - 1} r_S(X_t) = \sum_{t = 0}^3 r_S(X_t) \sim \text{DPH}_6(\bm{\pi}_S, \bm{T}_S), 
\end{align*}
with
\begin{align*}
    \bm{\pi}_S = (1/2, 0, 1/2, 0, 0, 0),\quad \bm{T}_S = \begin{pmatrix}
        0&1&0&0&0&0\\
        0&0&0&2/3&0&0\\
        0&0&0&1/3&0&1/3\\
        0&0&0&0&1&0\\
        0&0&0&0&0&0\\
        0&0&0&0&0&0
    \end{pmatrix},\quad \bm{t}_S = \begin{pmatrix}
        0\\1/3\\1/3\\0\\1\\1
    \end{pmatrix}.
\end{align*}
Applying \eqref{eq:MeanY}, we obtain $\bE[S] = 8/3$ and $\var(S) = 11/9$.

For the external branch length, we define
\begin{align}
    r_E(i(\bx)) = \bx_{n-1}, \label{eq:RewardTransformE}
\end{align}
which in the case $n=5$ yields
\begin{align}
    (r_E(1),\cdots,r_E(7)) = (5,3,2,1,1,0,0).\label{eq:rE5}
\end{align}
Applying \eqref{eq:MeanY} yields $\bE[E] = 10$ and $\var(E) = 2/3$.
\subsection{The multivariate discrete phase-type distribution}
Above, we introduced reward transforms as a means to recover distributional properties of single number summaries. Here, we generalize this by instead considering a random vector $(Y_1,\cdots,Y_m)$ where each $Y_j$ is formed by a reward transform. Formally, we define for each $j = 1,\cdots, m$ a reward transform $r_j : \{1,\cdots,p\} \to \bbn_0$ such that
\begin{align*}
    Y_j = \sum_{t=0}^{\tau - 1} r_j(X_t).
\end{align*}
Letting $\bm{R}$ denote the \textit{reward matrix} with 
\begin{align*}
    R_{ij} = r_j(i), \quad i = 1,\cdots,p,\quad j = 1,\cdots,m,
\end{align*}
we say that $\bm{Y} = (Y_1,\cdots,Y_m)$ has a multivariate discrete phase-type distribution with representation $(\bm{\pi}, \bm{T};\bm{R})$, and we write
\begin{align*}
    \bm{Y} \sim \text{MDPH}_p(\bm{\pi},\bm{T};\bm{R}).
\end{align*}

Proposition 5.7 in \cite{AzucenaPhD} provides matrix analytic formulae for the $l$th moments of $\bm{Y}$. In particular,  
\begin{align}
    \bE[Y_jY_k] &= \bm{\pi}\Bigl\{\bm{U} \bm{\Delta}(\bm{r}_j)\bm{U} \bm{\Delta}(\bm{r}_k) + \bm{U} \bm{\Delta}(\bm{r}_k)\bm{U} \bm{\Delta}(\bm{r}_j)-\bm{U} \bm{\Delta}(\bm{r}_j) \bm{\Delta}(\bm{r}_k)\Bigr\}\bm{e}, \label{eq:MeanYjYk} \\
    \cov(Y_j, Y_k) &= \bE[Y_jY_k] - \bE[Y_j]\bE[Y_k], \label{eq:covYjYk}
\end{align}
where $\bE[Y_j]$ is found using \eqref{eq:MeanY}. Again, the \texttt{PhaseTypeR} package is employed for the computations in the examples below.
\subsubsection*{Example: The joint distribution of E and S}
Let $n = 5$ and consider the ranked Kingman coalescent $\{X_t\}_{t=0,\cdots,4}$ with initial distribution and sub-transition probability matrix given in \eqref{eq:paramKingmanRcoal5}. Consider the rewards $r_S$ given in \eqref{eq:RewardTransformS} and $r_{E}$ given in \eqref{eq:RewardTransformE}. In this case, the reward matrix is given by
\begin{align*}
    \bm{R} = \begin{pmatrix}
        \bm{r}_E & \bm{r}_S 
    \end{pmatrix},
\end{align*}
where $\bm{r}_E$ and $\bm{r}_S$ are given in \eqref{eq:rE5} and \eqref{eq:rS5}. Using \eqref{eq:covYjYk} yields
\begin{align*}
    \cov(S,E) = 5/6.
\end{align*}
In Figure \ref{fig:MomentOfSnE}, we illustrate how the mean, variance, and covariance of $S$ and $E$ vary with increasing $n$.
\captionsetup{skip=5pt}
\begin{figure} [H]
    \centering
    \includegraphics[width=0.9\linewidth]{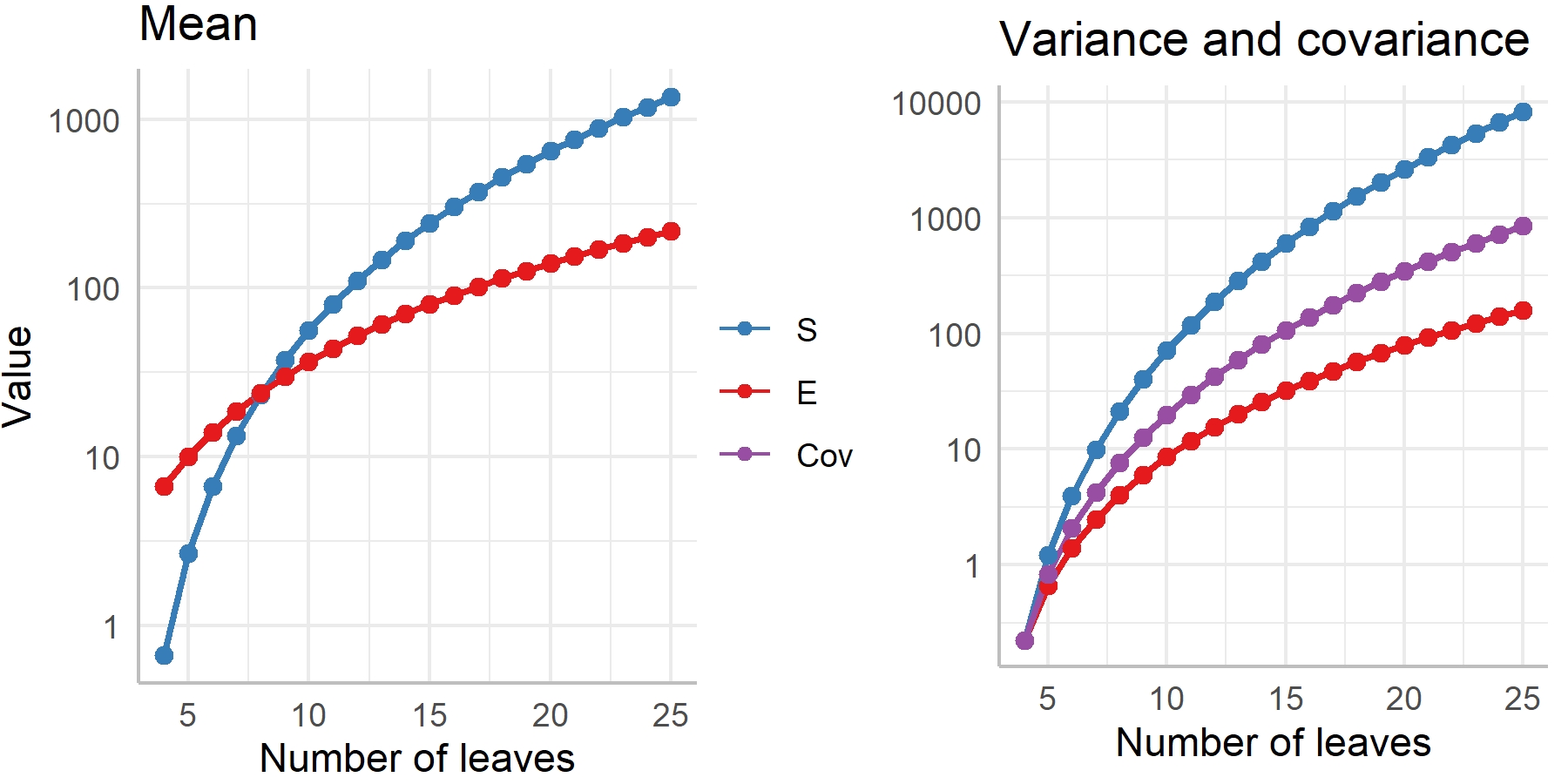}
    \caption{First and second order moments of $E$ and $S$ in the ranked Kingman coalescent for an increasing number of lineages.}
    \label{fig:MomentOfSnE}
    \vspace{-8pt}
\end{figure}
\subsection{The joint distribution of all non-fixed $\Fmat$ entries}
Consider the ranked coalescent $\{X_t\}_{t=0,\cdots,n-1}$ for a general number of lineages $n$. The non-fixed entries $F_{ij}$ are those below the second diagonal, $$\{F_{ij} : 1 \leq j \leq n-3, j+2 \leq i \leq n-1\}.$$
For each non-fixed position $(i,j)$, we define a reward transform $r_{ij}$ as
\begin{align*}
    r_{ij}(i(\bx)) = \bx_j \bm{1}\{t(\bx) = i - 1\},\quad \bx \in \mX_n. 
\end{align*}
With this definition, we get
\begin{align*}
    F_{ij} = \sum_{t=0}^{n-2} r_{ij}(i(X_t)).
\end{align*}
Collecting the reward transforms, we define
\begin{align}
    \bm{R}_F = \begin{pmatrix}
        \bm{r}_{31}&\bm{r}_{41}&\cdots&\bm{r}_{n-1,n-3}
    \end{pmatrix},\label{eq:RewardMatrixFij}
\end{align}
as the reward matrix, with non-fixed pairs ordered row-wise.
\subsubsection*{Example: Non-fixed entries for n = 5}
For $n=5$ entries $(3,1), (4,1), (4,2)$ are non-fixed, and the reward matrix is given by
\begin{align*}
    \bm{R}_F^\transpose = \begin{pmatrix}
        0&0&0&0&1&1&0\\
        0&0&0&0&1&0&0\\
        0&0&2&1&0&0&0
    \end{pmatrix}.
\end{align*}
In the Kingman ranked coalescent, the resulting mean vector $\bm{M}$ and covariance matrix $\bm{\Sigma}$ are found using \eqref{eq:paramKingmanRcoal5}, \eqref{eq:MeanY}, and \eqref{eq:covYjYk} as
\begin{align*}
    \bm{M} = \Bigl(2/3, 1/2, 3/2\Bigr)^\transpose, \quad \bm{\Sigma} = \begin{pmatrix}
        2/9&1/6&0\\
        1/6&1/4&1/12\\
        0&1/12&1/4
    \end{pmatrix}.
\end{align*}
The results of this section are summarized in the following theorem.
\begin{theorem}\label{Thm:MDPH}
    Let $n \geq 4$ and let $\bm{T},\bm{\pi}$ denote the sub transition probability matrix and initial distribution of the ranked coalescent.
    \begin{enumerate}[(i)]
        \item The external branch length $E$ follows a discrete phase-type distribution, with parameters obtainable from $\bm{T}, \bm{\pi}$ and reward transform given by \eqref{eq:RewardTransformE} using the \texttt{PhaseTypeR} package.
        \item Let $\bm{R}_{S,E}$ be the reward matrix given by columns \eqref{eq:RewardTransformS}, \eqref{eq:RewardTransformE}. The joint distribution of $S$ and $E$ is multivariate discrete phase-type, $$(S,E) \sim \text{MPDH}_{|\mX_n|-1}\bigl(\bm{\pi}_\text{RC},\bm{T}_\text{RC} ; \bm{R}_{S,E}\bigr).$$
        \item Let $\bm{R}_{F}$ be the reward matrix in \eqref{eq:RewardMatrixFij}. The joint distribution of the non-fixed entries $\{F_{ij} : 1 \leq j \leq n-3, j+2 \leq i \leq n-1\}$ is multivariate discrete phase-type, $$(F_{31},F_{41},\cdots,F_{n-1,n-3})\sim\text{MDPH}_{|\mX_n|-1}\bigl(\bm{\pi}_\text{RC},\bm{T}_\text{RC} ; \bm{R}_F\bigr).$$
    \end{enumerate}
\end{theorem}
Result (iii) along with \eqref{eq:MeanY}, \eqref{eq:covYjYk} is a generalization of Theorem 2 in \cite{SamPal}, since it does not assume the Kingman/Yule model as the underlying distribution, but allows for other demographic models, e.g. island models \cite{Wakeley}. While this provides a powerful theoretical framework, the size of the matrices involved grows exponentially with the number of lineages. Consequently, standard procedures for computing matrix inverses and products quickly become infeasible. In the following section, we exploit the feed-forward structure of the chain and circumvent these computational issues.\vspace{-5pt}

\section{Efficient matrix computations in feed-forward Markov chains}
The Markov chain of the ranked coalescent is one example of a process that encodes the structure of a coalescent. Another example is the block-counting process \cite{MPHSFS, PHpopgen} whose state space grows according to the partition function \cite{PartitionFunc}. This process governs the site frequency spectrum and is therefore of great interest in population genetics. Both processes' state spaces exhibit rapid growth with the number of lineages. This poses computational challenges when using phase-type theory for analysis, as necessary matrix operations such as multiplication and inversion scale cubically with the number of states. For example, $|\mX_{25}| = \text{Fib}(26) = 121,393$, making Theorem \ref{Thm:MDPH} practically infeasible using standard procedures. 
However, both admit a feed-forward structure that can be exploited to achieve substantial computational gains. Here, we demonstrate how this structure can be leveraged for the ranked coalescent. The methods are easily generalized to general feed-forward processes.

Inspired by the notion of tiers from Section $3$, we exploit the feed-forward structure by decomposing the sub-transition probability matrix as
\begin{align}
    \bm{T} = \begin{pmatrix}
        0&\bm{T}_{01}&0&\cdots&0\\
        0&0&\bm{T}_{12}&\cdots&0\\
        0&0&0&\ddots&0\\
        0&0&0&\cdots&\bm{T}_{n-3,n-2}\\
        0&0&0&\cdots&0
    \end{pmatrix}.\label{eq:SubTPMStructure}
\end{align}
Here, $\bm{T}_{i,i+1}$ is the rectangular matrix of transition probabilities from states in tier $i$ to states in tier $i+1$.

Since the chain is deterministically absorbed after exactly $n-1$ transitions, the sub-transition probability matrix $\bm{T}$ is nilpotent, with $\bm{T}^{n-1} = 0$. Consequently, the fundamental matrix admits the finite expansion
\begin{align*}
    \bm{U} = \bigl(\bm{I}-\bm{T}\bigr)^{-1} = \sum_{k=0}^\infty \bm{T}^k = \sum_{k=0}^{n-2} \bm{T}^k.
\end{align*}
This formulation provides a computationally much more efficient way to compute the vector-matrix product
\begin{align}
    \bm{\pi}\bm{U} = \sum_{k=0}^{n-2}\bm{\pi}\bm{T}^k,\label{eq:PiU}
\end{align}
which is essential for the computation of both the expectation and covariance.

To exploit the tiered structure of the feed-forward process even further, we partition the state space according to its tiers and define $$T_n^k = \{i(\bx) : t(\bx) = k\},\quad k = 0,\cdots,n-2$$ as the set of entries belonging to tier $k$. Using this notation, $\bm{\pi}\bm{T}^k$ can be computed iteratively, noting that $\bm{\pi}\bm{T}^k$ is $0$ everywhere except for entries belonging to $T_n^k$, as
\begin{align}
    \bigl(\bm{\pi}\bm{T}^{k}\bigr)_{i \in T_n^k} = \bigl(\bm{\pi}\bm{T}^{k-1}\bigr)_{i \in T_n^{k-1}} \bm{T}_{k-1,k}. \label{eq:LeftProducts}
\end{align}
Starting from $\bm{\pi}\bm{T}^0 = \bm{\pi}$, this iterative procedure efficiently computes $\bm{\pi}\bm{U}$ without ever forming large matrices or performing full inversions. Having obtained $\bm{\pi}\bm{U}$, the expected value of any reward transformed variable can be efficiently computed using \eqref{eq:MeanY}.

While this procedure enables efficient computations of expectations, variances and covariances of the non-fixed entries of $\Fmats$ are also of interest. Below, we derive an efficient method for computing the covariance between the non-fixed entries, noting that both $E,S$ can be written as affine transformations of these.

Let $F_{i,j}, F_{i',j'}$ be non-fixed entries. We recall the formula \eqref{eq:MeanYjYk} which can be decomposed into left and right factors as follows:
\begin{align*}
    \bE\bigl[F_{i,j}F_{i',j'}\bigr] &= \bm{\pi}\bigl\{\bm{U}\bm{\Delta}(\bm{r}_{ij})\bm{U}\bm{\Delta}(\bm{r}_{i'j'})+\bm{U}\bm{\Delta}(\bm{r}_{i'j'})\bm{U}\bm{\Delta}(\bm{r}_{ij})-\bm{U}\bm{\Delta}(\bm{r}_{ij})\bm{\Delta}(\bm{r}_{i'j'})\bigr\}\bm{e} \\
    &= \Bigl[\bigl(\bm{\pi}\bm{U}\bm{\Delta}(\bm{r}_{ij})\bigr)\bigl(\bm{U}\bm{\Delta}(\bm{r}_{i'j'})\bm{e}\bigr)\Bigr] + \Bigl[\bigl(\bm{\pi}\bm{U}\bm{\Delta}(\bm{r}_{i'j'})\bigr)\bigl(\bm{U}\bm{\Delta}(\bm{r}_{ij})\bm{e}\bigr)\Bigr] - \Bigl[\bm{\pi}\bm{U}\bm{\Delta}(\bm{r}_{ij})\bm{\Delta}(\bm{r}_{i'j'})\bm{e}\Bigr].
\end{align*}
Although the left factor of each summand has already been handled above, we still need to compute the right factors $\bm{U}\bm{\Delta}(\bm{r}_{ij})\bm{e}$. Here, the tiered structure can once again be exploited, noting that the entries in $\bm{T}^k\bm{\Delta}(\bm{r}_{ij})\bm{e}$ are $0$ except for those belonging to $T_n^{n-1-j-k}$. For $k = 1,\cdots,n-1-j$ we iteratively compute the relevant entries as
\begin{align}
    \bigl(\bm{T}^k\bm{\Delta}(\bm{r}_{ij})\bm{e}\bigr)_{i \in T_n^{n-1-j-k}} = \bm{T}_{n-1-j-k,n-j-k}\bigl(\bm{T}^{k-1}\bm{\Delta}(\bm{r}_{ij})\bm{e}\bigr)_{i \in T_n^{n-j-k}}. \label{eq:RightProducts}
\end{align}

When $n=25$ there are $253$ non-fixed entries. Computing all expectations, variances, and pairwise covariances requires $32,384$ total moment computations. Using standard methods, each computation involves matrix operations on the order of $\text{Fib}^3(26) \approx 10^{15}$, making the problem infeasible in practice. In contrast, the iterative procedures above reduce the total complexity to $$\sum_{k=0}^{n-3} |T_n^k| \cdot |T_n^{k+1}|,$$ yielding a dramatic computational improvement.

We illustrate the procedure in a small example, where we compute $\bE[F_{4,2}]$ and $\cov(F_{4,2},F_{5,1})$ in the ranked Kingman coalescent with $n=6$ lineages. Here, the tiers are given by $T_6^0 = \{1\}, T_6^1 = \{2\}, T_6^2 = \{3,4\}, T_6^3 = \{5,6,7,8\}, T_6^4 = \{9,10,11,12\}$. The matrices $\bm{T}_{01},\bm{T}_{12},\bm{T}_{23},\bm{T}_{34}$ are given below as
\begin{align*}
    \bm{T}_{01} &= \begin{pmatrix}
        1
    \end{pmatrix}, \quad \bm{T}_{12} = \begin{pmatrix}
        4/10&6/10
    \end{pmatrix}, \quad
    \bm{T}_{23} = \begin{pmatrix}
        3/6&0&3/6&0\\
        1/6&2/6&2/6&1/6
    \end{pmatrix},\quad
    \bm{T}_{34} = \begin{pmatrix}
        2/3&0&0&1/3\\
        1/3&1/3&0&1/3\\
        1/3&0&1/3&1/3\\
        0&1/3&1/3&1/3
    \end{pmatrix}.
\end{align*}
The reward transforms $\bm{r}_{4,2},\bm{r}_{5,1}$ are given by
\begin{align*}
    \bm{r}_{4,2} = (0,0,0,0,2,2,1,1,0,0,0,0)^\transpose, \quad \bm{r}_{5,1} = (0,0,0,0,0,0,0,0,1,0,0,0)^\transpose.
\end{align*}
Following the iteration in \eqref{eq:LeftProducts} we obtain
\begin{align*}
    (\bm{\pi}\bm{T}^0)_{1} &= 1, (\bm{\pi}\bm{T}^1)_{2} = 1 \cdot 1, (\bm{\pi}\bm{T}^2)_{3,4} = 1 \cdot \begin{pmatrix}
        4/10&6/10
    \end{pmatrix},\\
    (\bm{\pi}\bm{T}^3)_{5,6,7,8} &= \begin{pmatrix}
        4/10&6/10
    \end{pmatrix}\bm{T}_{23} = \begin{pmatrix}
        3/10&2/10&4/10&1/10
    \end{pmatrix}, \\
    (\bm{\pi}\bm{T}^4)_{9,10,11,12} &= \begin{pmatrix}
        3/10&2/10&4/10&1/10
    \end{pmatrix}\bm{T}_{34} = \begin{pmatrix}
        4/10&1/10&1/6&2/6
    \end{pmatrix}.
\end{align*}
Using this, we find that
\begin{align}
    \bm{\pi}\bm{U} = (1,1,4/10,6/10,3/10,2/10,4/10,1/10,4/10,1/10,1/6,2/6), \label{eq:piU}
\end{align}
so that
\begin{align*}
    \bE[F_{4,2}] &= \bm{\pi}\bm{U}\bm{r}_{4,2} = 3/10\cdot 2 + 2/10 \cdot 2 + 4/10 \cdot 1 + 1/10 \cdot 1 = 3/2.
\end{align*}
Similarly, we find that $\bE[F_{5,1}] = 2/5$.

For the covariance, we compute
\begin{align*}
    &\bigl(\bm{T}^0 \bm{\Delta}(r_{4,2})\bm{e}\bigr)_{i \in T_6^3} = (2,2,1,1)^\transpose, \bigl(\bm{T}^1 \bm{\Delta}(r_{4,2})\bm{e}\bigr)_{i \in T_6^2} = \bm{T}_{2,3}(2,2,1,1)^\transpose = (3/2,3/2)^\transpose \\
    &\bigl(\bm{T}^2 \bm{\Delta}(r_{4,2})\bm{e}\bigr)_{i \in T_6^1} = \bm{T}_{1,2}(3/2,3/2)^\transpose = 3/2, \bigl(\bm{T}^3 \bm{\Delta}(r_{4,2})\bm{e}\bigr)_{i \in T_6^0} = \bm{T}_{0,1}3/2 = 3/2,
\end{align*}
and thus
\begin{align*}
    \bm{U}\bm{\Delta}(\bm{r}_{4,2})\bm{e} = (3/2,3/2,3/2,3/2,2,2,1,1,0,0,0,0)^\transpose.
\end{align*}
Similarly, we find that
\begin{align*}
    \bm{U}\bm{\Delta}(\bm{r}_{5,1})\bm{e} = (4/10,4/10,1/2,1/3,2/3,1/3,1/3,0,1,0,0,0)^\transpose.
\end{align*}
Using \eqref{eq:piU} we find that 
\begin{align*}
    \bigl(\bm{\pi}\bm{U}\bm{\Delta}(\bm{r}_{4,2})\bigr)_{i \in T_6^3} &= (6/10, 4/10, 4/10, 1/10),\quad
    \bigl(\bm{\pi}\bm{U}\bm{\Delta}(\bm{r}_{5,1})\bigr)_{i \in T_6^4} = (4/10, 0, 0, 0).
\end{align*}
Lastly, we find $\bE\bigl[F_{4,2}F_{5,1}\bigr] = 2/3$, and therefore $\cov(F_{4,2},F_{5,1}) = 2/3 - 3/2\cdot2/5 = 1/15$.

These computational advances make it possible to evaluate expectations and covariances for moderately large trees, thereby increasing the size of the matrices that can be handled for any feed-forward Markov chain. In particular, they extend the tractable number of lineages for both the ranked coalescent and block-counting process. In the next section, we build on the results from the two previous sections, to construct three tests for assessing whether a sample of trees are consistent with a ranked coalescent model with a specified sub-transition matrix.\vspace{-5pt}

\section{Tests of neutrality}
As a final application of the Markov embedding, we examine the power of various tests for neutrality. Theorem \ref{Thm:MDPH} along with a multivariate central limit theorem for $\Fmats$, allow us to test whether or not a sample of ranked tree shapes follows a certain distribution with mean $\bm{M}_0$ and variance $\bm{\Sigma}_0$. A typical test of interest is the case where the null model corresponds to the neutral/Kingman coalescent.

In coalescent theory, a classic statistical approach for testing neutrality from observed molecular data is Tajima's D \cite{Tajima1989}. However, this test can be altered by other processes that affect branch length distributions \cite{SamPal}.
An alternative is to test neutrality using the ranked tree shape \cite{DrummondSuchard, SamPal}. Because the ranked coalescent ignores branch lengths, the tests developed below are based solely on these tree shapes, making them independent of variability in population size. The test in \cite{DrummondSuchard} relies on single summary statistics such as external branch length and the number of cherries, whereas \cite{SamPal} proposes a Hotelling's T-squared for assessing $H_0 : \bE[F] = M_0$. Similarly, we construct three topology-based tests using the $\Fmat$ representation. Theorem \ref{thm:CLTforSummaries} defines the test statistics and their approximate null distributions.
\begin{theorem}\label{thm:CLTforSummaries}
        Let $n, m \in \bbn$ and consider a sample $F_1,\cdots,F_m$ of $(n-1)\times(n-1)$ $\Fmats$, drawn from a distribution where the expectation of the non-fixed entries is denoted by $\bm{M}$ and the covariance by $\bm{\Sigma}$. For each $\Fmat$, let $E_i$ denote external branch length. Divide the range of $E$ into $K$ boxes and let $A_k$ denote the expected count within each box. Letting $O_k$ denote the observed counts in box $k$ we get approximately
        \begin{align}
            G_{E} = \sum_{k=1}^K O_k \log\Bigl(\frac{O_k}{A_k}\Bigr) \sim \chi^2(K-1). \label{eq:GoFE}
        \end{align}
        Define $\bar{F}_m$, $\bar{S}_m$, and $\bar{E}_m$ as the vector of averages of non-fixed entries, average sum of non-fixed entries, and average external branch length, respectively. We have
        \begin{align}
            \sqrt{m}\bigl(\bar{F}_m - \bm{M}\bigr) &\overset{\mathcal{D}}{\to}\mathcal{N}\bigl(\bm{0},\bm{\Sigma}\bigr)\label{eq:MultCLT},\\
            W_F := \sqrt{\frac{2m}{(n-2)(n-3)}}\Bigl[\bm{\Sigma}^{-1/2}\bigl(\bar{F}_m - \bm{M}\bigr)\Bigr]^\transpose \bm{e}&\overset{\mathcal{D}}{\to}\mathcal{N}(0,1) \label{eq:uniCLT},\\
            W_{SE} := \sqrt{\frac{m}{2}}\Bigl[\bm{\Sigma}_{SE}^{-1/2}\bigl((\bar{S}_m,\bar{E}_m)^\transpose-\bm{\mu}_{SE}\bigr)\Bigr]^\transpose \bm{e} &\overset{\mathcal{D}}{\to}\mathcal{N}(0,1),\quad \text{as }m\to \infty. \label{eq:SECLT}
        \end{align}
\end{theorem}\noindent
\begin{proof}
    The statistic in \eqref{eq:GoFE} is the goodness of fit and the asymptotic distribution follows from Wilks' theorem. \eqref{eq:MultCLT} follows directly from the multivariate central theorem, considering a vectorized version of the non-fixed entries of $\Fmats$. The results in \eqref{eq:uniCLT}, \eqref{eq:SECLT} follow from \eqref{eq:MultCLT} after applying the appropriate linear transformations.
\end{proof}
The distributional results in \eqref{eq:GoFE}, \eqref{eq:uniCLT}, and \eqref{eq:SECLT} can each be used to construct a test of neutrality. In Figure \ref{fig:powerplotall}, we estimate the power of these tests and compare the estimated power of each test to the estimated power of Hotelling's T-squared test from \cite{SamPal}. The powers were estimated by simulating $1000$ samples of $1000$ tree shapes with $n=25$ tips from the Blum-Fran\c{c}ois model \cite{Blum-Francois}, over a grid of parameter values $\beta \in (-1, 1]$.

All three tests attain the correct nominal level of $0.05$, as shown in Figure \ref{fig:powerplotall}. Additional validation of the distributional results in Theorem \ref{thm:CLTforSummaries} is provided in Figure \ref{fig:NullDistVerification} in the Appendix.

\begin{figure} [H]
    \centering
    \includegraphics[width=0.8\linewidth]{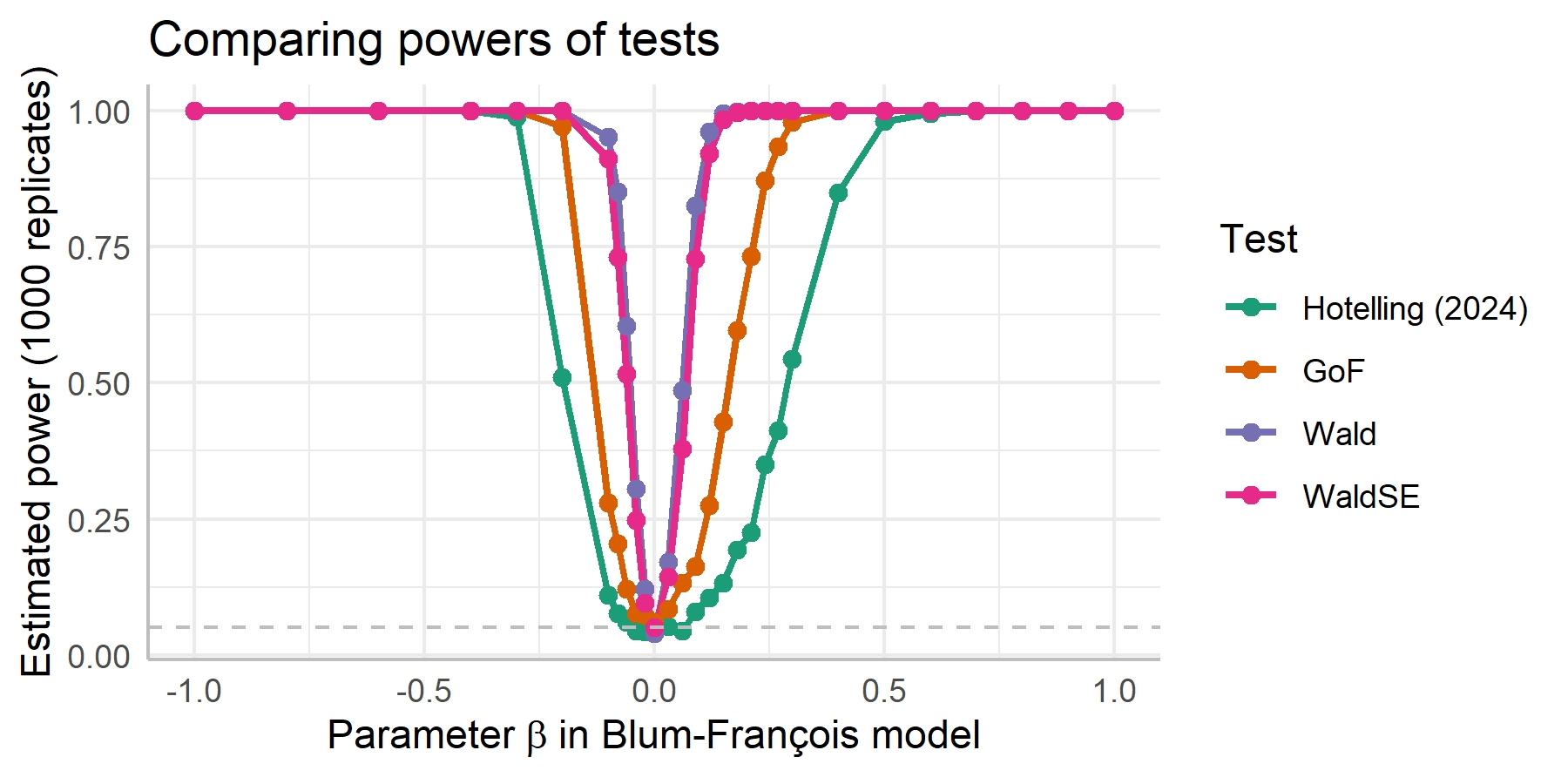}
    \caption{Estimated power of all tests including Hotelling's T-squared from \cite{SamPal}. The power is estimated using $1000$ replicates of $1000$ trees with $25$ leaves sampled from Blum-Fran\c{c}ois model with parameter $\beta$. We demonstrate that all the tests produced in this paper are more powerful.}
    \label{fig:powerplotall}
    \vspace{-5pt}
\end{figure}

\section{Discussion}
Markov chain embedding of genealogical processes is an emerging and powerful framework for understanding the stochastic structure of population genetics models \cite{MatrixAnalytics,MPHSFS,PHpopgen}. A key contribution of this paper is Theorem \ref{Thm:MarkovEmbeddingTheorem}, which demonstrates the Markov chain embeddability of ranked unlabelled trees. This adds a new perspective to the study of coalescent processes by focusing on tree topology. It also provides a state space much smaller than the tree space analyzed in \cite{SamPal}, as demonstrated in Theorem \ref{Thm:SizeOfSS}.

Using this smaller state space, we developed the recursive \texttt{ViTreebi} algorithm to efficiently compute all Fréchet means from a given distribution. We also generalized Theorem 2 in \cite{SamPal} using phase-type analysis, extending beyond the standard Kingman model and providing matrix analytical expressions for higher order moments. In Section $7$, we used these results to develop three statistical tests with greater power than Hotelling’s test in \cite{SamPal}.

The generality of the ranked coalescent suggests a wide potential beyond the classic Kingman framework. The island model \cite{Wakeley} is an important example: here, only lineages on the same island can coalesce at a given time. Consequently, the resulting tree topology is strongly dependent on initial conditions and migration rates. Extending the framework to include the island model is therefore a key direction for both theoretical development and potential applications.

While the ranked coalescent encodes the full structure of ancestral relationships, classical coalescent theory often represents such structures through the block-counting process (BCP) \cite{PHpopgen, MPHSFS}. The BCP records only the number of leaves per branch between coalescent events and is therefore coarser than the ranked coalescent, as every path through its associated jump process can be represented by one or more paths in the ranked coalescent. Because the state space is substantially smaller, computations based on the BCP are considerably faster. In particular, the external branch length can be derived from the BCP alone, whereas the sum of non-fixed entries depends on structural information not retained by the BCP. A detailed comparison of the two processes is provided in Section $9.3$ in the Appendix.

Although the framework's generality is one of its main advantages, extending the framework to include more coalescent models also highlights some of the limitations of the ranked coalescent. Despite the reduction in state space, the ranked coalescent still grows rapidly with the number of lineages. This creates challenges, both in terms of storing the transition probabilities and computing matrix inverses and products. In Section $6$, we developed methods for exploiting the feed-forward structure of the process, allowing us to extend our analysis to include the case of $n=25$ lineages.

Together, these results illustrate how the ranked coalescent provides a new and general framework for analyzing the topology of genealogical trees. The ranked coalescent offers a clear interpretation of the $\Fmats$ from \cite{SamPal,FmatPaper}, and is a practical tool for future applications involving small to medium sized trees.\medbreak\noindent
\small{\textbf{Acknowledgements}  We thank Julia Palacios for sharing the code on the Hotelling's T-squared test from \cite{SamPal}}.  \vspace{-5pt}

\section{Appendix}
\subsection{Efficient generation of the state space and transition matrix}
Practical implementation of the Markov embedding requires explicit construction of the state space and transition matrix. As the number of states grows exponentially with the number of lineages, efficient methods are essential for this construction. In this section, we introduce a binary representation of each state $\bx \in \mX_n$ and use it to construct the transition matrix under the Kingman model. 

Recall from Section $2$ that each state can be decomposed into lineages, where each lineage has a binary representation indicating its presence in each layer of the tree. Since each layer corresponds to a single coalescent event, at most one lineage is removed at each step, which is reflected as a transition from $1$ to $0$ in the binary representation. The last entry of each state gives the number of external lineages, while the remaining entries indicate the layers at which lineages are removed. Consequently, each state can be represented by tracking where lineages are removed and by noting the number of external lineages. Importantly, this representation allows all states and transitions to be generated using simple operations on binary vectors.

Formally, for $n \geq 3$, we define $d : \mX_n \to \{0,1\}^{n-2}\times \{0,\cdots,n\}$ by
\begin{align}
    d(\bx) = \Bigl(\bx - (\bx_{-1}^\transpose,0)^\transpose\Bigr)^+, \quad \bx \in \mX_n, \label{eq:DiffOperator}
\end{align}
where $\bx_{-1}^\transpose = (x_2,\cdots,x_{n-1})$ and where $\mathbf{y}^+$ is the vector with entries $\max(y_i,0)$. The last entry $d(\bx)_{n-1}$ gives the number of external lineages, while the remaining entries $d(\bx)_i, i = 1,\cdots,n-2$ indicate the layers at which a lineage is removed. For example, if $d(\bx)_i = 1$ for some $i = 1,\cdots,n-2$, then exactly one of the lineages is present at layer $i$, but not at layer $i+1$, and we refer to such an index as a \textit{decremental index}. In the following, we describe how these indices can be used to determine whether a transition is possible and, if possible, compute the probability.

Transitions from tier $t$ to tier $t+1$ occur when two lineages coalesce to form the new lineage $(0,\cdots,0,1,0,\cdots,0)^\transpose = \bm{e}_{n-2-t}$. By comparing the decremental indices between two states, we obtain the following necessary and sufficient condition for the feasibility of a transition.   

\begin{proposition}\label{prop:BinaryTransitions}
    Let $n\geq 4$ and $t \geq 0$. Consider the states $\bx,\mathbf{y} \in \mX_n$ of tiers $t, t+1$, respectively. A transition from $\bx$ to $\mathbf{y}$ is possible if and only if there exists a unique pair $(i,k)$ satisfying $i < k \leq n-1$ or $i=k=n-1$ such that
    \begin{align}
        d(\bx) - d(\mathbf{y}) + \bm{e}_{n-2-t} = \bm{e}_i + \bm{e}_k, \label{eq:BinaryTrans}
    \end{align}
    where $\bm{e}_j$ is the $j$th unit vector, and where the term $\bm{e}_{n-2-t}$ accounts for the difference in tiers between $\bx$ and $\mathbf{y}$.
\end{proposition}
\begin{proof}
    Assume that $\bbp(X_{t+1} = \mathbf{y} \mid X_t = \bx) > 0$, so that $\mathbf{y}$ is formed by the coalescence of two lineages in $\bx$. There are three possible cases:
    \begin{enumerate}
        \item \textbf{Two internal lineages coalesce}. Let $i < k < n-1$ denote the corresponding decremental indices. Then $d(\mathbf{y})$ inherits all decremental indices except $i$ and $k$, satisfying the condition \eqref{eq:BinaryTrans}.
        \item \textbf{One internal and one external coalesce.} Let $i < n-1$ be a decremental index and let $k=n-1$ correspond to an external lineage. Then $d(\mathbf{y})$ inherits all decremental indices except for $i$, and $\mathbf{y}$ has one less external lineage than $\bx$, again satisfying \eqref{eq:BinaryTrans}.
        \item \textbf{Two external lineages coalesce.} Set $i = k = n-1$. All decremental indices remain the same, but $\mathbf{y}$ has two fewer external lineages than $\bx$. Again, we see that \eqref{eq:BinaryTrans} is satisfied.
    \end{enumerate}
     Conversely, a pair $(i,k)$ satisfying \eqref{eq:BinaryTrans} uniquely determines which lineages of $\bx$ coalesce (up to non-uniqueness of the external lineages), ensuring that the transition to $\mathbf{y}$ is possible.
\end{proof}

The proof not only characterizes feasible transitions but also partitions them into three distinct scenarios. This forms the basis for computing transition probabilities directly from the binary representation, without reference to the underlying lineage decomposition, as formalized in the following theorem.

\begin{theorem}\label{thm:TPM}
    Consider the setup of Proposition \ref{prop:BinaryTransitions}, and assume that there exists a pair $(i,k)$ satisfying \eqref{eq:BinaryTrans}, and let $x_{\text{max}} = \max_i(x_i)$ be the total number of lineages in $\bx$. Under the Kingman model, the transition probability is given by
    \begin{align}
        \bbp(X_{t+1} = \mathbf{y} \mid X_t = \bx) = \binom{x_{\text{max}}}{2}^{-1}\cdot\begin{cases}
            1&\text{if } i < k < n-1\\[0.5em]
            x_{n-1}&\text{if } i < k = n-1\\[0.5em]
            \binom{x_{n-1}}{2}&\text{if }i=k=n-1
        \end{cases}.\label{eq:TransProbs}
    \end{align}
\end{theorem}
\begin{proof}
    Assuming the Kingman model, all lineages are equally likely to coalesce. The three cases identified in the proposition correspond directly to the type of coalescence.
    \begin{enumerate}
        \item \textbf{Two internal lineages} $i < k < n-1$: Each such coalescence leads to a unique state, giving a single possibility.
        \item \textbf{One internal, one external} $i < k = n-1$: All external lineages are identical, so each of the $x_{n-1}$ external lineages leads to the same transition.
        \item \textbf{Two external lineages} $i = k = n-1$: All external lineages are identical, and there are $\binom{x_{n-1}}{2}$ ways to select a pair, each resulting in the same transition.
    \end{enumerate}
    Lastly, we check that this gives a valid probability distribution. As $x_{\text{max}}$ gives the total number of lineages and $x_{n-1}$ is the number of external lineages, there are $x_{\text{max}}-x_{n-1}$ internal lineages. Therefore, there are $\binom{x_{\text{max}}-x_{n-1}}{2}$ distinct transitions corresponding to the first case. Similarly, we find that there are $x_{\text{max}}-x_{n-1}$ distinct transitions for the second case. Lastly, as all external lineages are identical, all transitions here lead to the same state. In total, we get
    \begin{align*}
        \binom{x_{\text{max}}-x_{n-1}}{2}+x_{n-1}(x_{\text{max}}-x_{n-1})+\binom{x_{n-1}}{2} = \binom{x_{\text{max}}}{2}.
    \end{align*}
\end{proof}\noindent
In practice, the binary representation is generated efficiently using a preloaded matrix of binary numbers. Proposition \ref{prop:BinaryTransitions} and Theorem \ref{thm:TPM} allow all transitions between consecutive tiers to be generated using simple vector operations. This forms the foundation for all further analysis. \vspace{-5pt}

\subsection{Proof for size of state space}
In this section, we provide a proof of Theorem \ref{Thm:SizeOfSS}. To show this, we first need a lemma that establishes a connection between states in $\mX_n$ and in $\mX_{n-1}$.
\begin{lemma}\label{lem:UniquePrev}
   Let $n\geq 4$. For all $\bx = (x_1,\cdots,x_{n-1})^T \in \mX_n$ except the initial state $(0,\cdots,0,n)$, there exists a state $\mathbf{y} = (y_1,\cdots,y_{n-2})^T \in \mX_{n-1}$ such that either
   \begin{enumerate}
       \item $\bx = (y_1,\cdots, y_{n-2}, y_{n-2})$, or
       \item $\bx = (y_1,\cdots, y_{n-2}, y_{n-2} - 1)$.
   \end{enumerate}
\end{lemma}
\begin{proof}
    Let $\bx = (x_1,\cdots,x_{n-1}) \in \mX_n$, and let $\mathbf{y} = (x_1,\cdots, x_{n-2})$. Considering a decomposition like \eqref{eq:X6decomp} of $\bx$, one obtains $\mathbf{y}$ by removing the last entry from each summand. This gives a decomposition of $\mathbf{y}$, showing that $\mathbf{y} \in \mX_{n-1}$. By the definition of $\mX_n$, the last entry of each non-initial state always satisfies $x_{n-1}\in\{x_{n-2},x_{n-2}-1\}$, as only one pair of lineages can coalesce at a time. This concludes the proof.
\end{proof}\noindent
Before continuing, we define for a fixed $n \geq 3$ and for each $j = 0,\cdots,n$ the set $$\mX_n^j := \{\bx\in \mX_n : x_{n-1} = j\},$$ of states having its last entry equal to $j$. Lemma \ref{lem:UniquePrev} now implies the following recursive structure regarding the size of $\mX_n^j$. 
\begin{corollary}\label{cor:RecursiveStates}
    Let $n \geq 4$. For each $j = 0,\cdots,n-2$ the size of $\mX_n^j$ can be calculated as
    \begin{align*}
        |\mX_n^j| = |\mX_{n-1}^j|+|\mX_{n-1}^{j+1}|.
    \end{align*}
\end{corollary}
\begin{proof}
    By Lemma \ref{lem:UniquePrev} define a mapping $\phi : \mX_n^j \to \mX_{n-1}^j \cup \mX_{n-1}^{j+1}$, by removing the last entry of each summand in a decomposition like \eqref{eq:X6decomp}. That is, $\phi(x_1,\cdots,x_{n-1}) = (x_1,\cdots,x_{n-2})$. The mapping is injective, since for any $\mathbf{y} \in \mX_{n-1}$, the last coordinate is uniquely determined as either $y_{n-2}$ or $y_{n-2}-1$. To prove surjectivity, let $\mathbf{y} = (y_1,\cdots,y_{n-2}) \in \mX_{n-1}^j$. Then $\bx = (y_1,\cdots,y_{n-2},y_{n-2}) \in \mX_n^j$. Similarly, if $\mathbf{y} \in \mX_{n-1}^{j+1}$, then $\bx = (y_1,\cdots,y_{n-2},y_{n-2}-1) \in \mX_n^j$. Hence $\phi$ is bijective, and thus $$|\mX_n^j| = |\mX_{n-1}^j \cup \mX_{n-1}^{j+1}| = |\mX_{n-1}^j| + |\mX_{n-1}^{j+1}|,$$ as $\mX_{n-1}^j \cap \mX_{n-1}^{j+1} = \emptyset$ by definition.
\end{proof}\noindent
As a consequence of this recursive structure, we get the following proposition.
\begin{proposition}
    Let $n \geq 3$. Then 
    \begin{enumerate}
        \item $|\mX_n^n| = 1, |\mX_n^{n-1}| = 0$
        \item $|\mX_n^j| = \text{Fib}(n-1-j)$ for all $j = 1,\cdots, n-2$
        \item $|\mX_n^0| = \text{Fib}(n-1)-1$.
    \end{enumerate}
\end{proposition}
\begin{proof}
    The two initial states of the ranked coalescent are $(0,\cdots,0,n)^\transpose, (0,\cdots,n-1,n-2)^\transpose$ with all other states having $x_{n-1} \leq n-3$. Therefore $|\mX_n^n| = 1,|\mX_n^{n-1}| = 0$.
    \medbreak \noindent
    We show points 2. and 3. by induction. For $n=3$, we have $\mX_3 = \{(0,3)^\transpose, (2,1)^\transpose\}$, so that $|\mX_3^2| = 0 = \text{Fib}(0), |\mX_3^1| = 1 = \text{Fib}(1)$, and $|\mX_3^0| = 0 = \text{Fib}(2)-1$.\\
    Inductive hypothesis: assume that for some $n \geq 3$, $|\mX_n^j| = \text{Fib}(n-1-j)$ for $1 \leq j \leq n-2$ and $|\mX_n^0| = \text{Fib}(n-1) - 1$. Then, by Corollary \ref{cor:RecursiveStates},
    \begin{align*}
        |\mX_{n+1}^j| = |\mX_{n}^j| + |\mX_{n}^{j+1}| = \text{Fib}(n-1-j)+\text{Fib}(n-1-j-1) = \text{Fib}(n-j),
    \end{align*}
    which proves point 2. Similarly,
    \begin{align*}
        |\mX_{n+1}^0| = |\mX_{n}^0| + |\mX_{n}^1| = \text{Fib}(n-1)-1 + \text{Fib}(n-2) = \text{Fib}(n) - 1.
    \end{align*}
\end{proof}\noindent
Using these results, we find that
\begin{align*}
        |\mX_n| &= 1 + \sum_{j=0}^n |\mX_n^j| = 1 + |\mX_n^0| + |\mX_n^n| + \sum_{j=1}^{n-1} |\mX_n^j|\\
        &= 1 + (\text{Fib}(n-1) - 1) + 1 + \sum_{j=1}^{n-1}\text{Fib}(n-1-j) = 1 + \sum_{k=0}^{n-1} \text{Fib}(k) = \text{Fib}(n+1).
\end{align*}
Here, the added $1$ accounts for the absorbing state MRCA. This concludes the proof of Theorem \ref{Thm:SizeOfSS}.
\subsection{The ranked block-counting process}
Traditional coalescent theory quantifies genetic diversity via the site frequency spectrum (SFS), whose dynamics are governed by the block-counting process (BCP) \cite{MPHSFS, PHpopgen}. Similarly to the ranked coalescent, the block-counting process describes the structure of the coalescent. In this section, we focus on the jump chain of the BCP which we denote by the ranked BCP and describe how it relates to the ranked coalescent. We then demonstrate that the external branch length, $E$, can be derived from the ranked BCP, while the sum of non-fixed entries, $S$, cannot. Analogously to Section $5$, this derivation enables the use of discrete-phase type theory for deriving the distribution of $E$ under neutral branching.  

\subsubsection{Comparing the ranked BCP to the ranked coalescent}
In this section, we formally introduce the BCP. The details are taken from \cite{MPHSFS}. The BCP tracks the number of branches $a_i$ that have $i$ descendants, where $i=1,\cdots,n$. With $n$ lineages, the state space of the BCP is given by
\begin{align}
    BCP_n := \Bigl\{\bm{a} = (a_1,\cdots,a_{n-1}) \in \mathbb{Z}_+^{n-1} : \sum_{i=1}^{n-1} i a_i = n\Bigr\}.\label{eq:statespaceBCP}
\end{align}
Possible transitions are
\begin{align*}
    (a_1,\cdots,a_{n-1}) \to (a_1,\cdots,a_i-1,\cdots,a_j-1,\cdots,a_{i+j}+1,\cdots,a_{n-1})
\end{align*}
with rate $a_ia_j$ whenever $a_i,a_j \geq 1$, and
\begin{align*}
    (a_1,\cdots,a_{n-1}) \to (a_1,\cdots,a_i-2,\cdots,a_{2i}+1,\cdots,a_{n-1})
\end{align*}
with rate $\binom{a_i}{2}$ for $a_i \geq 2$. The ranked BCP $\{B_t\}_{t=0,\cdots,n-2}$ is the jump process of this continuous time Markov chain. Similarly to the ranked coalescent, the ranked BCP starts in state $(n,0,\cdots,0)$ and then moves deterministically to state $(n-2,1,0,\cdots,0)$ before randomly moving to subsequent states until the MRCA is reached. Below, we illustrate the state space and transition probability matrix of the ranked BCP when $n=5$
\begin{figure}[H]
    \centering
    \scalebox{0.9}{
    \begin{tikzpicture}[->]
        \node[draw,circle,minimum size=1.5cm,inner sep=0pt, label={[shift={(0,-0.5)}]1}] at (-2,0) (state 1) {$\bigl(5,0,0,0\bigr)$};
    
        \node[draw,circle,minimum size=1.5cm,inner sep=0pt, label={[shift={(0,-0.5)}]2}] at (0,0) (state 2) {$\bigl(3,1,0,0\bigr)$};

        \node[draw,circle,minimum size=1.5cm,inner sep=0pt, label={[shift={(0,-0.5)}]3}] at (2,1) (state 3) {$\bigl(2,0,1,0\bigr)$};

        \node[draw,circle,minimum size=1.5cm,inner sep=0pt, label={[shift={(0,-0.5)}]4}] at (2,-1) (state 4) {$\bigl(1,2,0,0\bigr)$};

        \node[draw,circle,minimum size=1.5cm,inner sep=0pt, label={[shift={(0,-0.5)}]5}] at (4,1) (state 5) {$\bigl(1,0,0,1\bigr)$};

        \node[draw,circle,minimum size=1.5cm,inner sep=0pt, label={[shift={(0,-0.5)}]6}] at (4,-1) (state 6) {$\bigl(0,1,1,0\bigr)$};

        \node[draw,circle,minimum size=1.5cm,inner sep=0pt] at (6,0) (MRCA) {MRCA};

        \path (state 1) edge (state 2);
        \path (state 2) edge (state 3);
        \path (state 2) edge (state 4);
        \path (state 3) edge (state 5);
        \path (state 3) edge (state 6);
        \path (state 4) edge (state 5);
        \path (state 4) edge (state 6);
        \path (state 5) edge (MRCA);
        \path (state 6) edge (MRCA);

        \node at (10.5, 0) {$T = \begin{pmatrix}
            0&1&0&0&0&0\\
            0&0&1/2&1/2&0&0\\
            0&0&0&0&2/3&1/3\\
            0&0&0&0&1/3&2/3\\
            0&0&0&0&0&0\\
            0&0&0&0&0&0
        \end{pmatrix}$};
    \end{tikzpicture}}
    \caption{State space and sub-transition matrix of the ranked BCP in the case $n=5$. Notably, the ranked BCP contains $7$ states when MRCA is included, while the ranked coalescent contains $8$ (see Figure \ref{fig:Figure-Fmat}C). Compate to Figure 4 in \cite{MPHSFS}.}
    \label{fig:RankedBCPn5}
    \vspace{-10pt}
\end{figure}
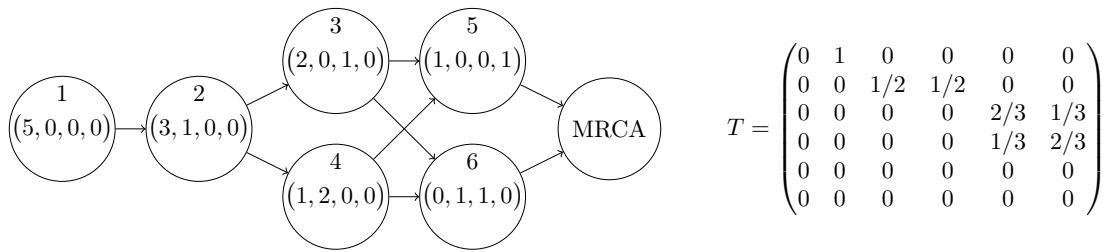

While the ranked BCP tracks the number of $i$-ton lineages, the ranked coalescent tracks the order in which the lineages coalesce. As a consequence, only singleton and non-singleton lineages can be distinguished with certainty in each individual state in the ranked coalescent. In contrast, the ranked BCP cannot differentiate between two doubletons, where one is formed early and one is formed much later. When $n=4$ this difference does not matter and the two processes coincide, which is illustrated in Figure \ref{fig:BCP_Rcoal_n4}.
\begin{figure}[!bt]
    \centering
    \scalebox{0.9}{
    \begin{tikzpicture}
        \node (CombTree) {\begin{tikzpicture}
        \draw [red] (-2, 0) -- (1, 0);
        \draw [red] (-2, 0) -- (-2, -1);
        \draw [red] (-3, -1) -- (-1, -1);
        \draw [red] (-3, -1) -- (-3, -2);
        \draw [red] (-1, -1) -- (-1, -3);
        \draw [red] (-4, -2) -- (-2, -2);
        \draw [red] (-4, -2) -- (-4, -3);
        \draw [red] (-2, -2) -- (-2, -3);
        \draw [red] (1, 0) -- (1, -3);

        \node[text width = 1cm, red] at (0,0.2) {$2$};
        \node[text width = 1cm, red] at (-1.75,-0.8) {$3$};
        \node[text width = 1cm, red] at (-2.75,-1.8) {$4$};
    \end{tikzpicture}};
        \node [right = 2.5cm of CombTree] (ForkTree) {\begin{tikzpicture}
        \draw [blue] (-2, 0) -- (1,0);
        \draw [blue] (-2, 0) -- (-2, -1);
        \draw [blue] (-3, -1) -- (-1, -1);
        \draw [blue] (-3, -1) -- (-3, -3);
        \draw [blue] (-1, -1) -- (-1, -3);
        \draw [blue] (1, 0) -- (1, -2);
        \draw [blue] (0, -2) -- (2,-2);
        \draw [blue] (0, -2) -- (0, -3);
        \draw [blue] (2, -2) -- (2, -3);

        \node[text width = 1cm, blue] at (-0.6,0.2) {$2$};
        \node[text width = 1cm, blue] at (-2.4,-0.8) {$3$};
        \node[text width = 1cm, blue] at (0.6,-1.8) {$4$};
    \end{tikzpicture}};

    \node [below = 0.75cm of CombTree] (RankedCoal) {\begin{tikzpicture}[->]
        \node at (-2, 0) (state 1) {$\begin{pmatrix}
            0\\0\\4
        \end{pmatrix}$};

        \node at (0, 0) (state 2) {$\begin{pmatrix}
            0\\3\\2
        \end{pmatrix}$};

        \node at (2, 1.8) (state 3) {$\begin{pmatrix}
            2\\1\\1
        \end{pmatrix}$};

        \node at (2, -1.8) (state 4) {$\begin{pmatrix}
            2\\1\\0
        \end{pmatrix}$};

        \node [anchor=north] at (4, -0.5) (MRCA) {MRCA};

        \path (state 1) edge (state 2);
        \path [red] (state 2) edge (state 3);
        \path [blue] (state 2) edge (state 4);
        \path [red] (state 3) edge (MRCA);
        \path [blue] (state 4) edge (MRCA);
    \end{tikzpicture}};

    \node [below = 0.75cm of ForkTree] (RankedBCP) {\begin{tikzpicture}[->]
        \node[draw,circle,minimum size=1.5cm,inner sep=0pt, label={[shift={(0,-0.5)}]1}] at (-2,0) (state 1) {$\bigl(4,0,0\bigr)$};

        \node[draw,circle,minimum size=1.5cm,inner sep=0pt, label={[shift={(0,-0.5)}]2}] at (0,0) (state 2) {$\bigl(2,1,0\bigr)$};

        \node[draw,circle,minimum size=1.5cm,inner sep=0pt, label={[shift={(0,-0.5)}]3}] at (2,2) (state 3) {$\bigl(1,0,1\bigr)$};

        \node[draw,circle,minimum size=1.5cm,inner sep=0pt, label={[shift={(0,-0.5)}]4}] at (2,-2) (state 4) {$\bigl(0,2,0\bigr)$};

        \node[draw,circle,minimum size=1.5cm,inner sep=0pt] at (4, 0) (MRCA) {MRCA};

        \path (state 1) edge (state 2);
        \path [red] (state 2) edge (state 3);
        \path [blue] (state 2) edge (state 4);
        \path [red] (state 3) edge (MRCA);
        \path [blue] (state 4) edge (MRCA);
    \end{tikzpicture}};
    \node (HeaderNode) at ([shift={(-0.5cm,2cm)}] CombTree) {};
    \node (A) at ([shift={(0cm,0cm)}] HeaderNode) {\Large{\bf A:} The two trees with $n=4$ leaves};
    \node (B) at ([shift={(0cm,-4.2cm)}] HeaderNode) {\Large{\bf B:} The ranked coalescent for $n=4$};
    \node (C) at ([shift={(8.25cm,-4.2cm)}] HeaderNode) {\Large{\bf C:} The ranked BCP for $n=4$};
    \end{tikzpicture}}
    \caption{{\bf A: } The two ranked unlabelled trees with $n=4$ leaves.{\bf B:} The transition diagram in the ranked coalescent. {\bf C:} The transition diagram in the ranked BCP. We note that the two transition diagrams coincide.}
    \label{fig:BCP_Rcoal_n4}
    \vspace{-15pt}
\end{figure}

As already hinted in Figure \ref{fig:RankedBCPn5}, the ranked coalescent contains an additional state in the last tier before MRCA compared to the ranked BCP when $n=5$. This is due to the fact that both states $(2,1,1,0)^\transpose, (2,1,0,0)^\transpose$ can be formed by the coalescence of a doubleton and a singleton. This shows that a state in the ranked BCP, here $(0,1,1,0)$ can be represented by multiple states in the ranked coalescent. 

Similarly, a state in the ranked coalescent also has the potential to have multiple representations in the ranked BCP. As an example, consider the decomposition of the state $(2,1,1,0,0)^\transpose \in \mX_6$
\begin{align*}
    \begin{pmatrix}
        2\\1\\1\\0\\0
    \end{pmatrix} = \underbrace{\begin{pmatrix}
        1\\0\\0\\0\\0
    \end{pmatrix}}_a + \underbrace{\begin{pmatrix}
        1\\1\\1\\0\\0
    \end{pmatrix}}_b.\vspace{-5pt}
\end{align*}
Here, lineage $a$ is either a tripleton or quadrupleton while $b$ is either a doubleton or a tripleton. Thus, this state can correspond to either $(0,0,2,0,0)$ or $(0,1,0,1,0)$ in the ranked BCP. 

Although states in both processes can often be represented in several ways for the other process, the state space of the ranked BCP grows much slower with $n$ compared to the ranked coalescent. The size of $BCP_n$ is equal to the partition function $p(n) \approx \frac{1}{4n\sqrt{3}}\exp(\pi\sqrt{2n/3})$ \cite{PartitionFunc}, which almost grows exponentially, albeit at a much lower rate than $\text{Fib}(n)$. In Figure \ref{fig:compBCPRcoal}, we compare the sizes of the state spaces.

The paths in the ranked BCP provide a more coarse description of tree topology compared to the complete description by the paths in the ranked coalescent. As a consequence, some measures of tree balance, such as the sum of non-fixed entries, are not available from the ranked BCP. However, the much smaller state space of the ranked BCP allows for more efficient computations. In the following section, we derive a phase-type distribution of the external branch length based only on the ranked BCP alone.

\subsubsection{External branch length in the ranked BCP}
Similarly to the ranked coalescent, the ranked BCP is a discrete-time Markov chain with an absorbing state. Therefore, we can apply phase-type methods in this situation as well. In Section 3, we introduced the notion of external lineages as those equal to $(0,\cdots,1,\cdots,1)^\transpose$. These are lineages that have yet to coalesce and in context of the ranked BCP they are therefore identical to singletons. This allows us to define a reward transform $r_{E}^{BCP} : \{1,2,\cdots,|BCP_n|\} \to \bbn$ by
\begin{align}
    r_{E}^{BCP}(i) = \#\{\text{singletons in state }i\} \label{eq:RewardTransformEBCP}.
\end{align}
Based on this, we can find the distribution of $E$ given a sub-transition probability matrix $\bm{T}_{BCP}$ of the ranked BCP.
\begin{figure}[!bt]
    \centering
    \includegraphics[width=0.7\linewidth]{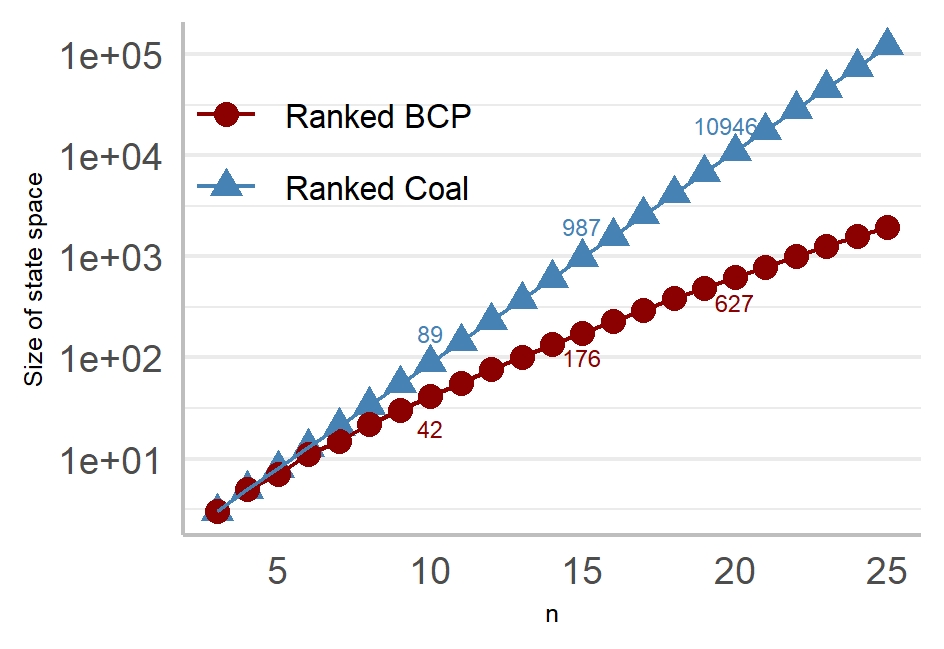}
    \caption{Comparing the sizes of the two state spaces. The state space of the ranked coalescent grows at a much faster rate compare to the ranked BCP.}
    \label{fig:compBCPRcoal}
\end{figure}
As an example, we consider the case $n=4$. Here, the initial distribution $\bm{\pi}_{BCP} = (1,0,0,0)$ and the sub-transition probability matrix
\begin{align*}
    \bm{T}_{BCP} = \begin{pmatrix}
        0&1&0&0\\
        0&0&2/3&1/3\\
        0&0&0&0\\
        0&0&0&0
    \end{pmatrix}.
\end{align*}
We find $|BCP_4| = 4$ and $(r_{E}^{BCP}(1),r_{E}^{BCP}(2),r_{E}^{BCP}(3),r_{E}^{BCP}(4)) = (4,2,1,0)$. Consequently,
\begin{align*}
    E = \sum_{t=0}^{2} r_{E}^{BCP}(B_t) \sim \text{DPH}_7\bigl(\bm{\pi}_{E}^{BCP}, \bm{T}_{E}^{BCP}\bigr),
\end{align*}
with
\begin{align*}
    \bm{\pi}_{E}^{BCP} = (1,0,0,0,0,0,0), \quad \bm{T}_{E}^{BCP} = \begin{pmatrix}
        0&1&0&0&0&0&0\\
        0&0&1&0&0&0&0\\
        0&0&0&1&0&0&0\\
        0&0&0&0&1&0&0\\
        0&0&0&0&0&1&0\\
        0&0&0&0&0&0&2/3\\
        0&0&0&0&0&0&0
    \end{pmatrix}.
\end{align*}
We find $E = 7$ if $B_2 = 3$ and $E = 6$ if $B_2 = 4$, which occur with probability $2/3$ and $1/3$, respectively. Computing $\bm{\pi}_{E}\bm{T}_{E}^{m-1}\bm{t}_{E}$ with $m = 6,7$ confirms these results.

The much smaller state space of the ranked BCP allows us to employ the \texttt{PhaseTypeR} \cite{PhaseTypeR} package directly. In Figure \ref{fig:EnullDist}, we illustrate the distribution of $E$ with $n=25$ lineages based on the ranked Kingman BCP.
\begin{figure}[H]
    \centering
    \includegraphics[width=0.7\linewidth]{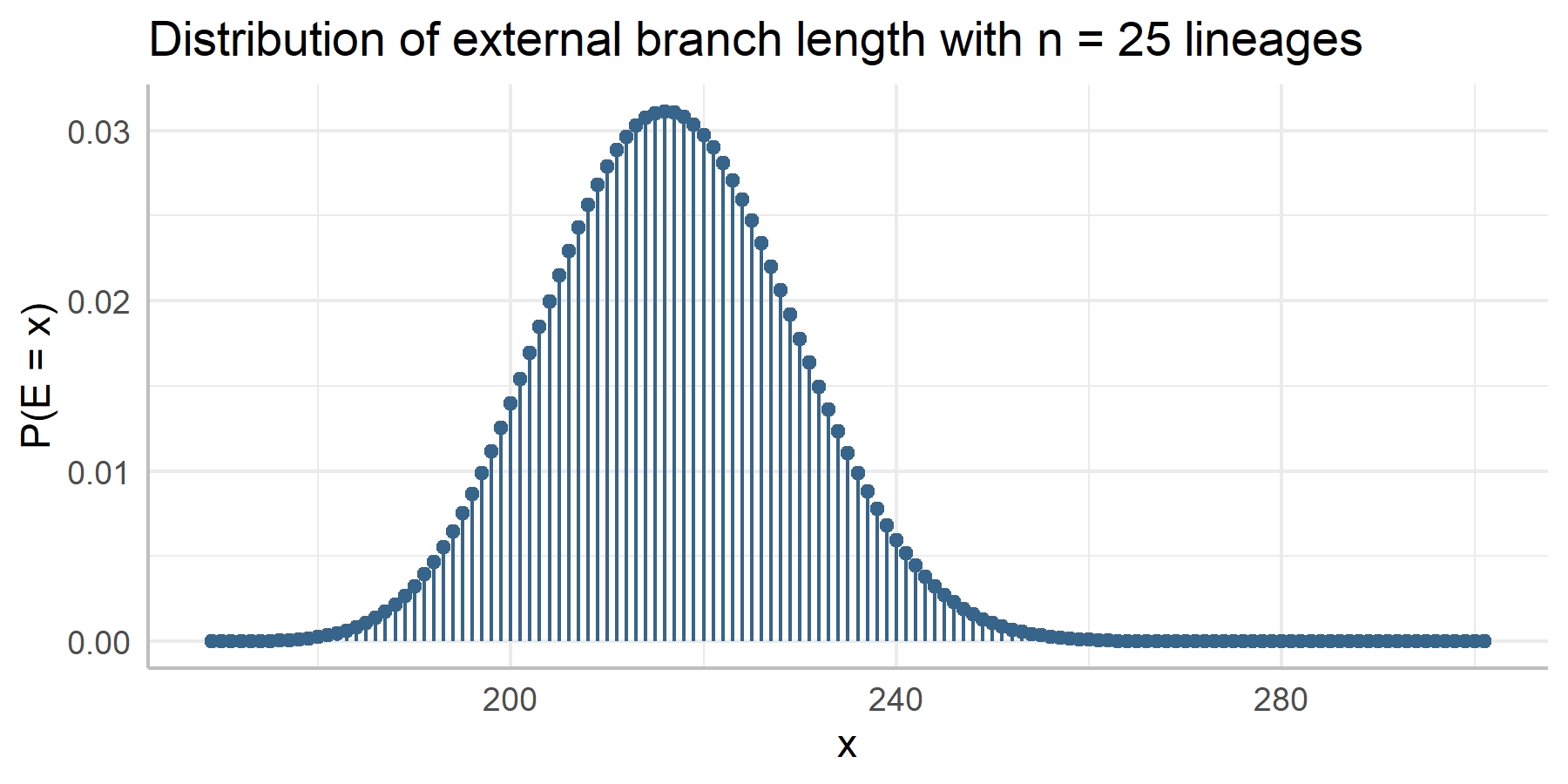}
    \caption{Distribution of external branch length with $n=25$ lineages. Probabilities are found using the \texttt{PhaseTypeR} package.}
    \label{fig:EnullDist}
    \vspace{-10pt}
\end{figure}
\subsection{Verification of correct null distributions}
In Section $7$, we introduced three statistical tests for neutrality based on a sample of $\Fmats$. The tests were shown to be more powerful compared to the previous test suggest in \cite{SamPal}. For each test, we demonstrated the correct level of $0.05$ based on samples from the ranked Kingman coalescent. Here, we further demonstrate the correctness by plotting the distribution of each of the test statistics based on samples from the ranked Kingman coalescent. 
\begin{figure}[H]
    \centering
    \includegraphics[width=0.8\linewidth]{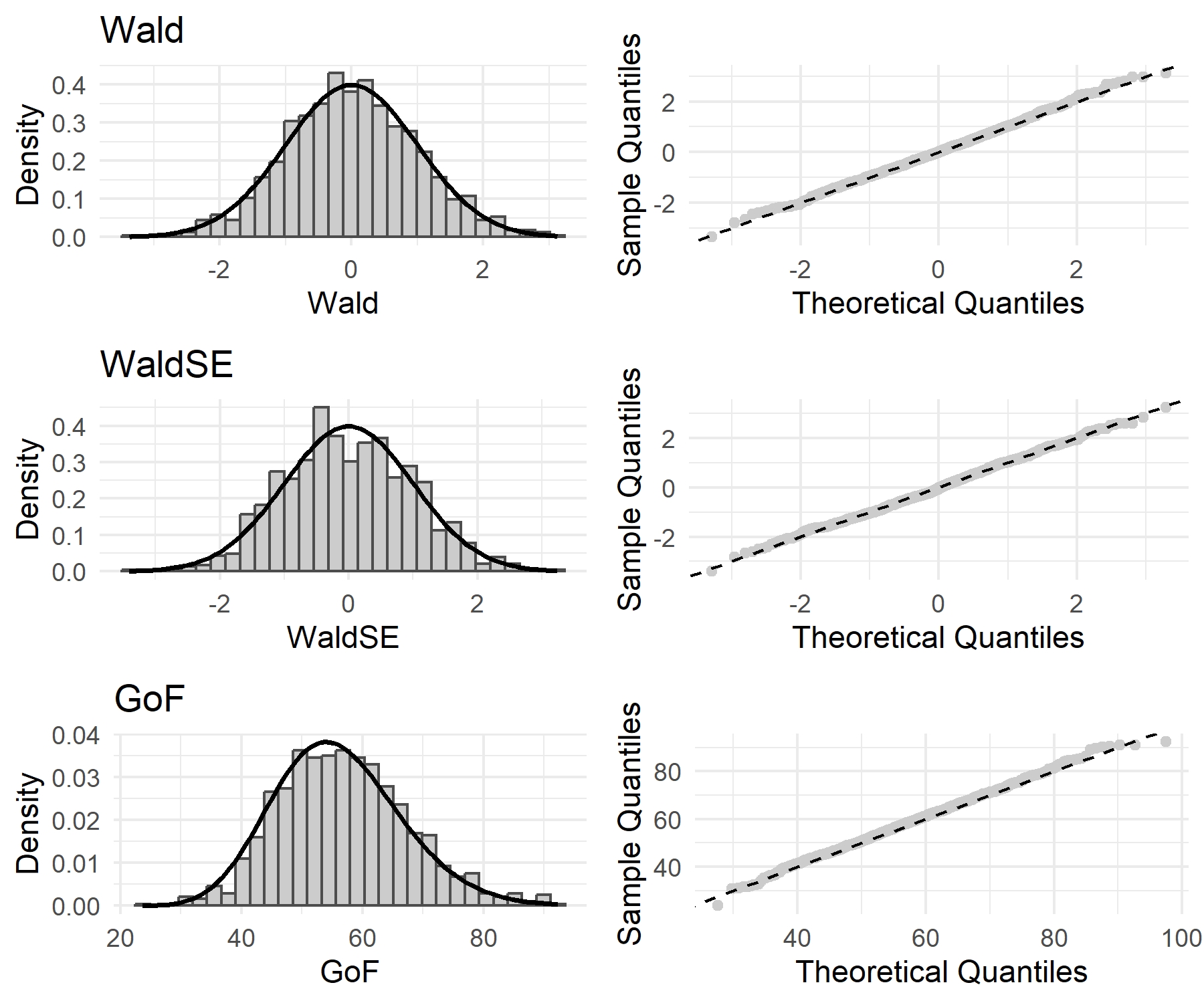}
    \caption{Histogram and quantile plots based on $1000$ test statistics simulated from the ranked Kingman coalescent. The null distributions are superimposed.}
    \label{fig:NullDistVerification}
\end{figure}

\printbibliography
\end{document}